\documentclass{llncs}

\usepackage{amssymb,amsmath,arydshln,color,algorithm,algpseudocode,multirow}
\usepackage[dvips]{graphicx}
\usepackage{wrapfig}

\newcommand{\texcomment}[1]{}

\newcommand{\comment}[1]{\textcolor{red}{{#1}}}

\newcommand{\todo}{\comment{TODO}}

\newenvironment{reftheorem}[1]{\begin{trivlist}\item[\hskip
      \labelsep{\bf Theorem #1.}]\it}{\end{trivlist}}

\newcommand{\rest}{\mathbin{\restriction}}

\newcommand{\ve}[1]{\boldsymbol{#1}}
\newcommand{\compcirc}[2]{({#1})\mathbin{\rhd}{#2}}

\newcommand{\comptblm}[3]{{#1}\mathbin{\rhd}_{#3}{#2}}

\newcommand{\comptbl}[2]{{#1}\mathbin{\rhd}{#2}}

\newcommand{\smallcolon}{\hspace*{-2pt}:\hspace*{-2pt}}
\newcommand{\aset}[1]{\{{#1}\}}

\newcommand{\aseq}[1]{\langle{#1}\rangle}

\newcommand{\tbl}[0]{{\it tbl}}

\newcommand{\brest}[3]{{#1}_{#2}({#3})}

\newcommand{\nd}[0]{\alpha}
\newcommand{\ndb}[0]{\beta}

\newcommand{\reduce}[2]{\xleftarrow{({#1},{#2})}}
\newcommand{\reductions}[1]{\leftarrow^{#1}}

\newcommand{\rows}[1]{{\it \#rows}({#1})}
\newcommand{\cols}[1]{{\it \#cols}({#1})}

\newcommand{\hd}[2]{{\it hd}({#1},{#2})}
\newcommand{\tl}[2]{{\it tl}({#1},{#2})}

\newcommand{\splitouts}[3]{{\it outs}_{{#1},{#2}}({#3})}

\newcommand{\nonsplitvals}[1]{{\it nsv}({#1})}

\newcommand{\match}[2]{\mathbin{\succeq_{#1,#2}}}
\newcommand{\eqtbl}[0]{\mathbin{\simeq}}

\newcommand{\safe}[3]{{#1}\mathbin{\sim}{#2,#3}}
\newcommand{\dettbl}[2]{{#1}\mathbin{\sim}{#2}}

\newcommand{\pvar}{p}
\newcommand{\svar}{k}
\newcommand{\rvar}{r}
\newcommand{\ovar}{o}
\newcommand{\var}{v}

\newcommand{\nodes}[1]{{\it nodes}({#1})}

\newcommand{\store}{\nu}

\newcommand{\eval}[3]{{#1}_{#2}({#3})}

\newcommand{\dist}[4]{\mu_{#1}({#2},{#3},{#4})}

\newcommand{\cnt}[5]{\#_{#1}({#2},{#3},{#4},{#5})}

\newcommand{\plusone}[1]{\({#1}\)\(+\)\(1\)}

\newcommand{\parcomp}[2]{{#1}||{#2}}

\newcommand{\phip}{\Phi_P}
\newcommand{\phiskel}[1]{\Phi_{\textit{Sk}_{#1}}}

\newcommand{\skel}[1]{\textit{Sk}_{#1}}

\newcommand{\phiioeq}{\Phi_{\textsc{io}}}
\newcommand{\philr}{\Phi_{\textsc{lr}}}

\newcommand{\cvar}{C}

\newcommand{\testset}{\textit{tset}}

\newcommand{\tst}{\ve{b_\pvar}\ve{b_\svar}}

\newcommand{\tstprime}{\ve{b_\pvar}\ve{b_\svar'}}

\newcommand{\rspace}{\mathcal{R}}

\newcommand{\depins}[1]{\textit{deps}({#1})}

\newcommand{\nodecount}{\gamma}

\newcommand{\nosol}{\textsf{nosol}}

\newcommand{\sol}{\textsf{sol}}

\newcommand{\success}{\textsf{success}}

\newcommand{\cex}{\textsf{cex}}

\newcommand{\assign}{\textbf{:=}}

\newcommand{\mult}{\mathbin{\cdot}}

\newcommand{\legname}{name}

\newcommand{\legtime}{time}

\newcommand{\legsize}{size}

\newcommand{\legnumrands}{rds}

\newcommand{\legmaxtpc}{mtc}

\newcommand{\nodesel}[2]{{#1}\aseq{#2}}

\newcommand{\addappendix}[2]{#1}

\begin{document}

\title{Compositional Synthesis of Leakage Resilient Programs}
\author{Arthur Blot\inst{1} \and Masaki Yamamoto\inst{2} \and Tachio Terauchi\inst{3}}
\institute{ENS Lyon, Lyon, France\\ \email{arthur.blot@ens-lyon.fr}
\and Nagoya University, Nagoya, Japan\\ \email{yamamoto-m@sqlab.jp}
  \and JAIST, Nomi, Japan\\\email{terauchi@jaist.ac.jp}}

\maketitle

\begin{abstract}
  A promising approach to defend against side channel attacks is to
  build programs that are {\em leakage resilient}, in a formal sense.
  One such formal notion of leakage resilience is the {\em
    $n$-threshold-probing model} proposed in the seminal work by Ishai
  et al.~\cite{DBLP:conf/crypto/IshaiSW03}.  In a recent
  work~\cite{DBLP:conf/cav/EldibW14}, Eldib and Wang have proposed a
  method for automatically synthesizing programs that are leakage
  resilient according to this model, for the case $n=1$.  In this
  paper, we show that the $n$-threshold-probing model of leakage
  resilience enjoys a certain compositionality property that can be
  exploited for synthesis.  We use the property to design a synthesis
  method that efficiently synthesizes leakage-resilient programs in a
  compositional manner, for the general case of $n > 1$.  We have
  implemented a prototype of the synthesis algorithm, and we
  demonstrate its effectiveness by synthesizing leakage-resilient
  versions of benchmarks taken from the literature.
\end{abstract}

\section{Introduction}


Side channel attacks are well recognized as serious threat to the
security of computer systems.  Building a system that is resilient to
side channel attacks is a challenge, particularly because there are
many kinds of side channels (such as, power, timing, and
electromagnetic radiation) and attacks on them.  In an effort to
establish a principled solution to the problem, researchers have
proposed formal definitions of resilience against side channel
attacks, called {\em leakage
  resilience}~\cite{DBLP:conf/crypto/IshaiSW03,DBLP:conf/eurocrypt/IshaiPSW06,DBLP:conf/eurocrypt/FaustRRTV10,DBLP:conf/crypto/GoldwasserR10,DBLP:conf/focs/GoldwasserR12,DBLP:conf/asiacrypt/BalaschFGV12}.
The benefit of such formal models of side-channel-attack resilience is
that a program proved secure according to a model is guaranteed to be
secure against all attacks that are permitted within the model.

The previous research has proposed various notions of leakage
resilience.  In this paper, we focus on the {\em $n$-threshold-probing
  model} proposed in the seminal work by Ishai et
al.~\cite{DBLP:conf/crypto/IshaiSW03}.  Informally, the model says
that, given a program represented as a Boolean circuit, the adversary
learns nothing about the secret by executing the program and observing
the values of at most $n$ nodes in the circuit
(cf.~Section~\ref{sec:prelim} for the formal definition).
\texcomment{Attacks such as differential power
  analysis~\cite{DBLP:conf/crypto/KocherJJ99} can be formally
  prevented by the model because they are reducible to computing
  intermediate values in the program.}  The attractive features of the
model include its relative simplicity, and the relation to {\em
  masking}, a popular countermeasure technique used in security
practice.  More precisely, the security under the
$n$-threshold-probing model is equivalent to the security under {\em
  $n\textsuperscript{th}$-order masking}~\cite{DBLP:conf/ches/RivainP10}, and often, the
literature uses the terminologies
interchangeably~\cite{DBLP:conf/asiacrypt/BalaschFGV12,DBLP:conf/fse/CoronPRR13,DBLP:conf/tacas/EldibWS14,DBLP:conf/cav/EldibW14,DBLP:conf/eurocrypt/BartheBDFGS15,DBLP:journals/iacr/BartheBDFG15}.
Further, as recently shown by Duc et
al.~\cite{DBLP:conf/eurocrypt/DucDF14}, the security under the model
also implies the security under the {\em noisy} leakage
model~\cite{DBLP:conf/eurocrypt/ProuffR13} in which the adversary
obtains information from every node with a probabilistically
distributed noise.

In a recent work, Eldib and Wang~\cite{DBLP:conf/cav/EldibW14} have
proposed a synthesis method that, given a program represented as a
circuit, returns a functionally equivalent circuit that is leakage
resilient according to the $n$-threshold-probing model, for the case
$n=1$ (i.e., the adversary observes only one node).  The method is a
constraint-based algorithm whereby the constraints expressing the
necessary conditions are solved in a CEGAR (counterexample-guided
abstraction refinement) style.  In this work, we extend the synthesis
to the general case where $n$ can be greater than $1$.  Unfortunately,
naively extending (the monolithic version of) their algorithm to the
case $n > 1$ results in a method whose complexity is double
exponential in $n$, leading to an immediate roadblock.\footnote{Their
  paper~\cite{DBLP:conf/cav/EldibW14} also shows a compositional
  algorithm.  However, compositionality becomes non-trivial when $n >
  1$ because then the adversary can observe nodes from the different
  components of the composition.}  As we show empirically in
Section~\ref{sec:impl}, the cost is highly substantial, and the naive
monolithic approach fails even for the case $n = 2$ on reasonably
simple examples.

Our solution to the problem is to exploit a certain {\em
  compositionality} property admitted by the leakage resilience model.
We state and prove the property formally in Theorems~\ref{thm:parcomp}
and \ref{thm:seqmultinput}.  Roughly, the compositionality theorems
say that composing $n$-leakage-resilient circuits results in an
$n$-leakage-resilient circuit, under the condition that the randoms in
the components are disjoint.  The composition property is quite
general and is particularly convenient for synthesis.  It allows a
compositional synthesis method which divides the given circuit into
smaller sub-circuits, synthesizes $n$-leakage-resilient versions of
them, and combines the results to obtain an $n$-leakage-resilient
version of the whole.  The correctness is ensured by using disjoint
randoms in the synthesized sub-circuits.  Our approach is an
interesting contrast to the approach that aims to achieve
compositionality without requiring the disjointness of the component's
randoms, but instead at the cost of additional randoms at the site of
the
composition~\cite{DBLP:conf/fse/CoronPRR13,DBLP:journals/iacr/BartheBDFG15}.

We remark that the compositionality is not at all obvious and quite
unexpected.  Indeed, at first glance, $n$-leakage resilience for each
individual component seems to say nothing about the security of the
composed circuit against an adversary who can observe the nodes of
multiple different components in the composition.  To further
substantiate the non-triviality, we remark that the compositionality
property is quite sensitive, and for example, it fails to hold if the
bounds are relaxed even slightly so that the adversary makes at most
$n$ observations within each individual component but the total number
of observations is allowed to be just one more than $n$
(cf.~Example~\ref{ex:seqtight}).

To synthesize $n$-leakage-resilient sub-circuits, we extend the
monolithic algorithm from \cite{DBLP:conf/cav/EldibW14} to the case
where $n$ can be greater than $1$.  We make several improvements to
the baseline algorithm so that it scales better for the case $n > 1$
(cf.~Section~\ref{sec:mono}).  We have implemented a prototype of our
compositional synthesis algorithm, and experimented with the
implementation on benchmarks taken from the literature.  We summarize
our contributions below.
\begin{itemize}
\item A proof that the $n$-threshold-probing model of leakage
  resilience is preserved under certain circuit compositions
  (Section~\ref{sec:comp}).

\item A compositional synthesis algorithm for the leakage-resilience
  model that utilizes the compositionality property
  (Section~\ref{sec:alg}).

\item Experiments with a prototype implementation of the synthesis
  algorithm (Section~\ref{sec:impl}).

\end{itemize}

The rest of the paper is organized as follows.
Section~\ref{sec:prelim} introduces preliminary definitions and
notations, including the formal definition of the
$n$-threshold-probing model of leakage resilience.
Section~\ref{sec:comp} states and proves the compositionality
property.  Section~\ref{sec:alg} describes the compositional synthesis
algorithm.  We report on the experience with a prototype
implementation of the algorithm in Section~\ref{sec:impl}, and discuss
related work in Section~\ref{sec:related}.  We conclude the paper in
Section~\ref{sec:concl}.  \addappendix{Appendix}{The extended report~\cite{extendedreport}}
contains the omitted proofs and extra materials.

\section{Preliminaries}

\label{sec:prelim}

We use boldface font for finite sequences.  For example, $\ve{b} =
b_1,b_2,\dots,b_n$.  We adopt the standard convention of the
literature~\cite{DBLP:conf/tacas/EldibWS14,DBLP:conf/cav/EldibW14,DBLP:conf/eurocrypt/BartheBDFGS15,DBLP:journals/iacr/BartheBDFG15}
and assume that a {\em program} is represented as an acyclic Boolean
circuit.\footnote{In the implementation described in
  Section~\ref{sec:impl}, following the previous
  works~\cite{DBLP:conf/tacas/EldibWS14,DBLP:conf/cav/EldibW14}, we
  convert the given C program into such a form.}  We assume the usual
Boolean operators, such as XOR gates $\oplus$, AND gates $\wedge$, OR
gates $\vee$, and NOT gates $\neg$.

A program has three kinds of inputs, {\em secret inputs} (often called
{\em keys}) ranged over by $\svar$, {\em public inputs} ranged over by
$\pvar$, and {\em random inputs} ranged over by $\rvar$.  Informally,
secret inputs contain the secret bits to be protected from the
adversary, public inputs are those that are visible and possibly given
by the adversary, and random inputs contain bits that are generated
uniformly at random (hence, it may be more intuitive to view randoms
as not actual ``inputs'').

Consider a program $P$ with secret inputs $\svar_1,\dots,\svar_x$.  In
the $n$-threshold-probing model of leakage resilience, we prepare {\em
  \plusone{n}-split shares} of each $\svar_i$ (for $i \in
\aset{1,\dots,x}$):
\[
\rvar_{i,1}, \ \dots, \ \rvar_{i,n}, \ \svar_i \oplus (\bigoplus_{j=1}^n \rvar_{i,j})
\]
where each $\rvar_{i,j}$ is fresh.  Note that the split shares sum to
$\svar_i$, and (assuming that $\rvar_{i,j}$'s are uniformly
independently distributed) observing up to $n$ many shares reveals no
information about $\svar_i$.  Adopting the standard terminology of the
literature~\cite{DBLP:conf/crypto/IshaiSW03}, we call the circuit that
outputs such split shares the {\em input encoder} of $P$.  The leakage
resilience model also requires an {\em output decoder}, which sums the
split shares at the output.  More precisely, suppose $P$ has $y$ many
\plusone{n}-split outputs $\ve{\ovar_1},\dots,\ve{\ovar_y}$ (i.e.,
$|\ve{\ovar_i}| = n+1$ for each $i \in \aset{1,\dots,y}$).  Then, the
output decoder for $P$ is, for each $\ve{\ovar_i} =
\ovar_{i,1},\dots,\ovar_{i,n+1}$, the circuit $\bigoplus_{j=1}^{n+1}
\ovar_{i,j}$.  For example, Fig.~\ref{fig:ex} shows a
\plusone{2}-split circuit with the secret inputs $\svar_1,\svar_2$,
public inputs $\svar_1,\svar_2$, random inputs
$\rvar_1,\rvar_2,\rvar_3,\rvar_4$, and two outputs.  Note that the
input encoder (the region \textbf{Input Encoder}) introduces the
randoms, and the output decoder (the region \textbf{Output
  Decoder}) sums the output split shares.
\begin{wrapfigure}{r}{8cm}
\vspace{-0.4cm}
\begin{center}
\includegraphics[width=8cm]{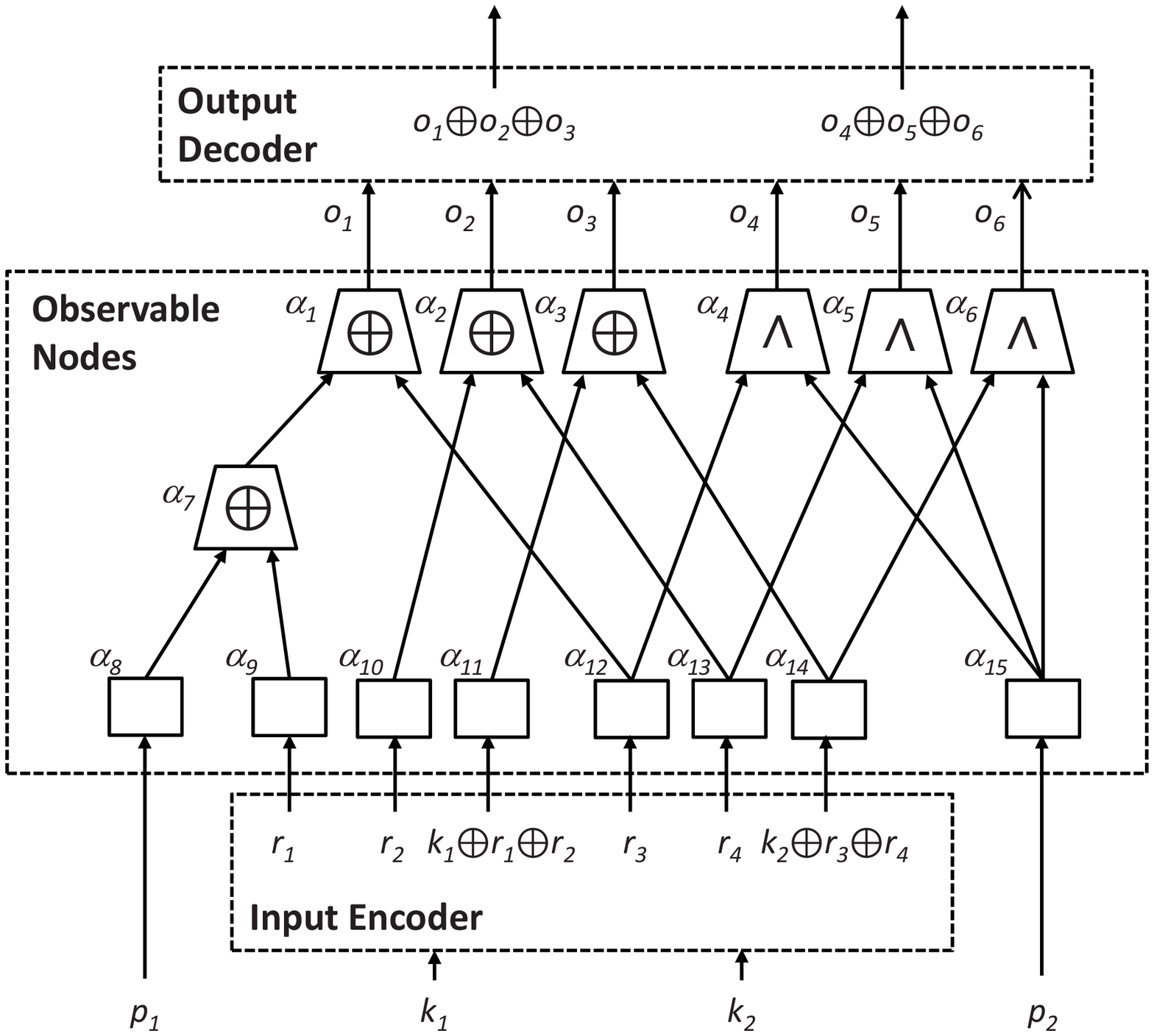}
\caption{Example of a 2-leakage-resilient circuit}
 \label{fig:ex}
\end{center}
\vspace{-1cm}
\end{wrapfigure}

We associate a unique label (ranged over by $\nd$) to every gate of
the circuit.  We call the {\em nodes} of $P$, $\nodes{P}$, to be the
set of labels in $P$ excluding the gates internal to the input encoder
and the output decoder part of $P$ (but including the outputs of the
input encoder and the inputs to the output decoder).  Intuitively,
$\nodes{P}$ are the nodes that can be observed by the adversary. For
example, in Fig.~\ref{fig:ex}, the observable nodes are the ones in
the region \textbf{Observable Nodes} labeled $\nd_1,\dots,\nd_{15}$.

Let $\store$ be a mapping from the inputs of $P$ to $\aset{0,1}$.  Let
$\ve{\nd} = \nd_1,\dots,\nd_n$ be a vector of nodes in $\nodes{P}$.
We define the {\em evaluation}, $\eval{\store}{P}{\ve{\nd}}$, to be
the vector $b_1,\dots,b_n \in \aset{0,1}^n$, such that each $b_i$ is
the valuation of the node $\nd_i$ when evaluating $P$ under $\store$.
For example, let $P'$ be the circuit from Fig.~\ref{fig:ex}.  Let
$\store$ map $\pvar_1, \rvar_1,\svar_2$ to $0$, and the others to $1$.
Then, $\eval{\store}{P'}{\nd_{11},\nd_1} = 0,1$.

Let us write $\store[\ve{\var} \mapsto \ve{b}]$ for the store $\store$
except that each $\var_i$ is mapped to $b_i$ (for $i \in \aset{1,\dots
  m}$) where $\ve{\var} = \var_1,\dots,\var_m$ and $\ve{b} =
b_1,\dots,b_m$.  Let $P$ be a circuit with secret inputs $\ve{\svar}$,
public inputs $\ve{\pvar}$, and random inputs $\ve{\rvar}$.  For
$\ve{b_\pvar} \in \aset{0,1}^{|\pvar|}$, $\ve{b_\svar} \in
\aset{0,1}^{|\svar|}$, $\ve{\nd} \in \nodes{P}^*$, and $\ve{b_\nd} \in
\aset{0,1}^{|\ve{\nd}|}$, let
\(
\cnt{P}{\ve{b_\pvar}}{\ve{b_\svar}}{\ve{\nd}}{\ve{b_\nd}} =
|\aset{\ve{b} \in \aset{0,1}^{|\ve{\rvar}|} \mid
  \eval{\store[\ve{\rvar}\mapsto\ve{b}]}{P}{\ve{\nd}} = \ve{b_\nd}}|
\)
where $\store =
\aset{\ve{\pvar}\mapsto\ve{b_\pvar},\ve{\svar}\mapsto\ve{b_\svar}}$.
We define $\dist{P}{\ve{b_\pvar}}{\ve{b_\svar}}{\ve{\nd}}$ to be the finite map
from each $\ve{b_\nd}\in \aset{0,1}^{|\ve{\nd}|}$ to
$\cnt{P}{\ve{b_\pvar}}{\ve{b_\svar}}{\ve{\nd}}{\ve{b_\nd}}$.  We
remark that $\dist{P}{\ve{b_\pvar}}{\ve{b_\svar}}{\ve{\nd}}$, when
normalized by the scaling factor $2^{-|\ve{\rvar}|}$, is the joint
distribution of the values of the nodes $\ve{\nd}$ under the public
inputs $\ve{b_\pvar}$ and the secret inputs $\ve{b_\svar}$.

Roughly, the $n$-threshold-probing model of leakage resilience says
that, for any selection of $n$ nodes, the joint distribution of the
nodes' values is independent of the secret.  Formally, the
leakage-resilience model is defined as follows.
\begin{definition}[Leakage Resilience]
  \em Let $P$ be an \plusone{n}-split circuit with secret inputs
  $\ve{\svar}$, public inputs $\ve{\pvar}$, and random inputs
  $\ve{\rvar}$.  Then, $P$ is said to be {\em leakage-resilient under
    the $n$-threshold-probing model} (or, simply {\em
    $n$-leakage-resilient}) if for any $\ve{b_\pvar} \in
  \aset{0,1}^{|\ve{\pvar}|}$, $\ve{b_\svar} \in
  \aset{0,1}^{|\ve{\svar}|}$, $\ve{b_\svar}' \in
  \aset{0,1}^{|\ve{\svar}|}$, and $\ve{\nd} \in \nodes{P}^n$,
  $\dist{P}{\ve{b_\pvar}}{\ve{b_\svar}}{\ve{\nd}} =
  \dist{P}{\ve{b_\pvar}}{\ve{b_\svar}'}{\ve{\nd}}$.
\end{definition}
We remark that, above, $\ve{\rvar}$ includes all randoms introduced by
the input encoder as well as any additional ones that are not from the
input encoder, if any.  For instance, in the case of the circuit from
Fig.~\ref{fig:ex}, the randoms are $\rvar_1,\rvar_2,\rvar_3,\rvar_4$
and they are all from the input encoder.

\begin{wrapfigure}{r}{12.5em}
\vspace{-1cm}
\begin{center}
\includegraphics[width=10em]{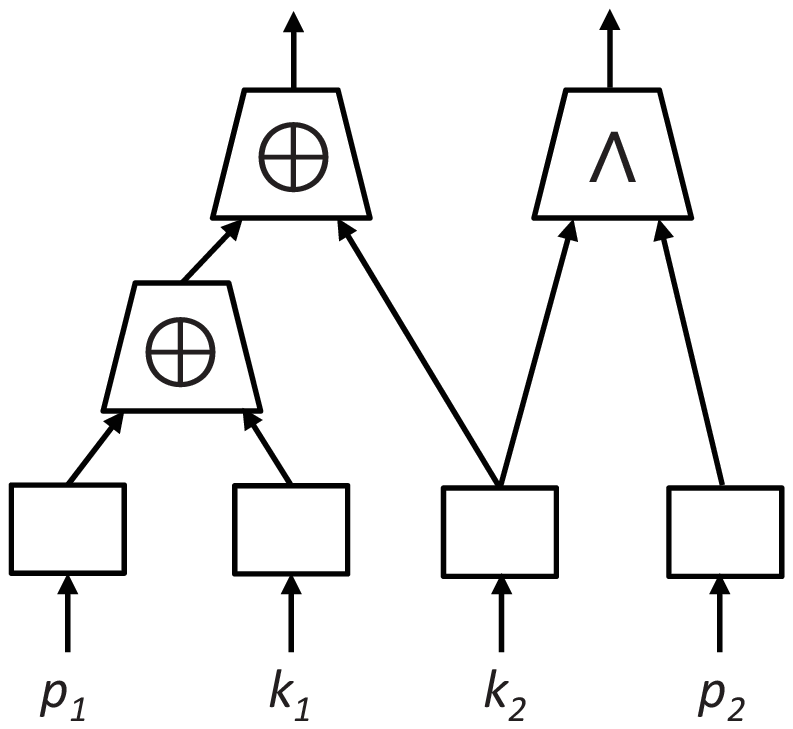}
\caption{A circuit computing $(\pvar_1\oplus\svar_1\oplus\svar_2, \svar_2\wedge\pvar_2)$}
 \label{fig:exorig}
\end{center}
\vspace{-1cm}
\end{wrapfigure}
Informally, the $n$-threshold-probing model of leakage resilience says
that the attacker learns nothing about the secret by executing the
circuit and observing the values of up to $n$ many internal gates and
wires, excluding those that are internal to the input
encoder and the output decoder.

We say that a circuit is {\em random-free} if it has no randoms.  Let
$P$ be a random-free circuit with public inputs $\ve{\pvar}$ and
secret inputs $\ve{\svar}$, and $P'$ be a circuit with public inputs
$\ve{\pvar}$, secret inputs $\ve{\svar}$, and randoms $\ve{\rvar}$.
We say that $P'$ is {\em IO-equivalent} to $P$ if for any
$\ve{b_\svar} \in \aset{0,1}^{|\ve{\svar}|}$, $\ve{b_\pvar} \in
\aset{0,1}^{|\ve{\pvar}|}$, and $\ve{b_\rvar} \in
\aset{0,1}^{|\ve{\rvar}|}$, the output of $P$ when evaluated under
$\store = \aset{\ve{\pvar} \mapsto \ve{b_\pvar},\ve{\svar} \mapsto
  \ve{b_\svar}}$ is equivalent to that of $P'$ when evaluated under
$\store[\ve{\rvar} \mapsto \ve{b_\rvar}]$.  We formalize the synthesis
problem.
\begin{definition}[Synthesis Problem]
  \em Given $n > 0$ and a random-free circuit $P$ as input, the {\em synthesis problem} is the problem of
  building a circuit $P'$ such that 1.) $P'$ is IO-equivalent to $P$,
and 2.) $P'$ is $n$-leakage-resilient.
\end{definition}

An important result by Ishai et al.~\cite{DBLP:conf/crypto/IshaiSW03}
is that any random-free circuit can be converted to an IO-equivalent
leakage-resilient form.
\begin{theorem}[\cite{DBLP:conf/crypto/IshaiSW03}]
\label{thm:ishai03}
  For any random-free circuit $P$, there exists an
  $n$-leakage-resilient circuit that is IO-equivalent to $P$.
\end{theorem}
While the result is of theoretical importance, the construction is
more of a proof-of-concept in which every gate is transformed
uniformly, and the obtained circuits can be quite unoptimized (e.g.,
injecting excess randoms to mask computations that do not depend on
secrets).  The subsequent research has proposed to construct more
optimized leakage-resilient circuits
manually~\cite{DBLP:conf/ches/RivainP10,DBLP:conf/fse/CoronPRR13}, or
by automatic synthesis~\cite{DBLP:conf/cav/EldibW14}.  The latter is
the direction of the present paper.

\begin{example}
Consider the random-free circuit $P$ shown in Fig.~\ref{fig:exorig}
which outputs
$(\pvar_1\oplus\svar_1\oplus\svar_2,\svar_2\wedge\pvar_2)$.  Let $P'$
be the circuit from Fig.~\ref{fig:ex}.  It is easy to see that $P'$
is IO-equivalent to $P$.  Also, it can be shown that $P'$ is
$2$-leakage resilient.  Therefore, $P'$ is a $2$-leakage-resilient
version of $P$.
\end{example}

\begin{remark}
  The use of the input encoder and the output decoder is
  unavoidable.  It is easy to see that the input encoder is needed.
  Indeed, without it, one cannot even defend against an
  one-node-observing attacker as she may directly observe the secret.
  To see that the output decoder is also required, consider an
  one-output circuit without the output decoder and let $n$ be the
  fan-in of the last gate before the output.  Then, assuming that the
  output depends on the secret, the circuit cannot defend itself
  against an $n$-nodes-observing attacker as she may observe the
  inputs to the last gate.
\end{remark}

\begin{remark}
\label{rem:multenc}
In contrast to the previous
works~\cite{DBLP:conf/fse/CoronPRR13,DBLP:journals/iacr/BartheBDFG15}
that implicitly assume that each secret is encoded (i.e., split in
\plusone{n} shares) by only one input encoder, we allow a secret to be
encoded by multiple input encoders.  The relaxation is important in
our setting because, as remarked before, the compositionality results
require disjointness of the randoms in the composed components.
\end{remark}

\noindent{\bf Split and Non-Split Inputs / Outputs.}
We introduce terminologies that are convenient when describing the
compositionality results in Section~\ref{sec:comp}.  We use the term
{\em split inputs} to refer to the \plusone{n} tuples of wires to
which the \plusone{n}-split (secret) inputs produced by the input
encoder (i.e., the pair of triples
$\rvar_1,\rvar_2,\svar_1\oplus\rvar_1\oplus\rvar_2$ and
$\rvar_3,\rvar_4,\svar_2\oplus\rvar_3\oplus\rvar_4$ in the example of
Fig.~\ref{fig:ex}) are passed, and use the term {\em non-split inputs}
to refer to the wires to which the original inputs before the split
(i.e., $\svar_1$ and $\svar_2$ in Fig.~\ref{fig:ex}) are passed.  We
define {\em split outputs} and {\em non-split outputs} analogously.
Roughly, the split inputs and outputs are the inputs and outputs
of the attacker-observable part of the circuit (i.e., the region {\bf
  Observable Nodes} in Fig.~\ref{fig:ex}), whereas the non-split
inputs and outputs are those of the whole circuit with the
input encoder and the output decoder.

\section{Compositionality of Leakage Resilience}
\label{sec:comp}

\begin{wrapfigure}{r}{15em}
\vspace{-1.2cm}
\begin{center}
\includegraphics[width=7.5em]{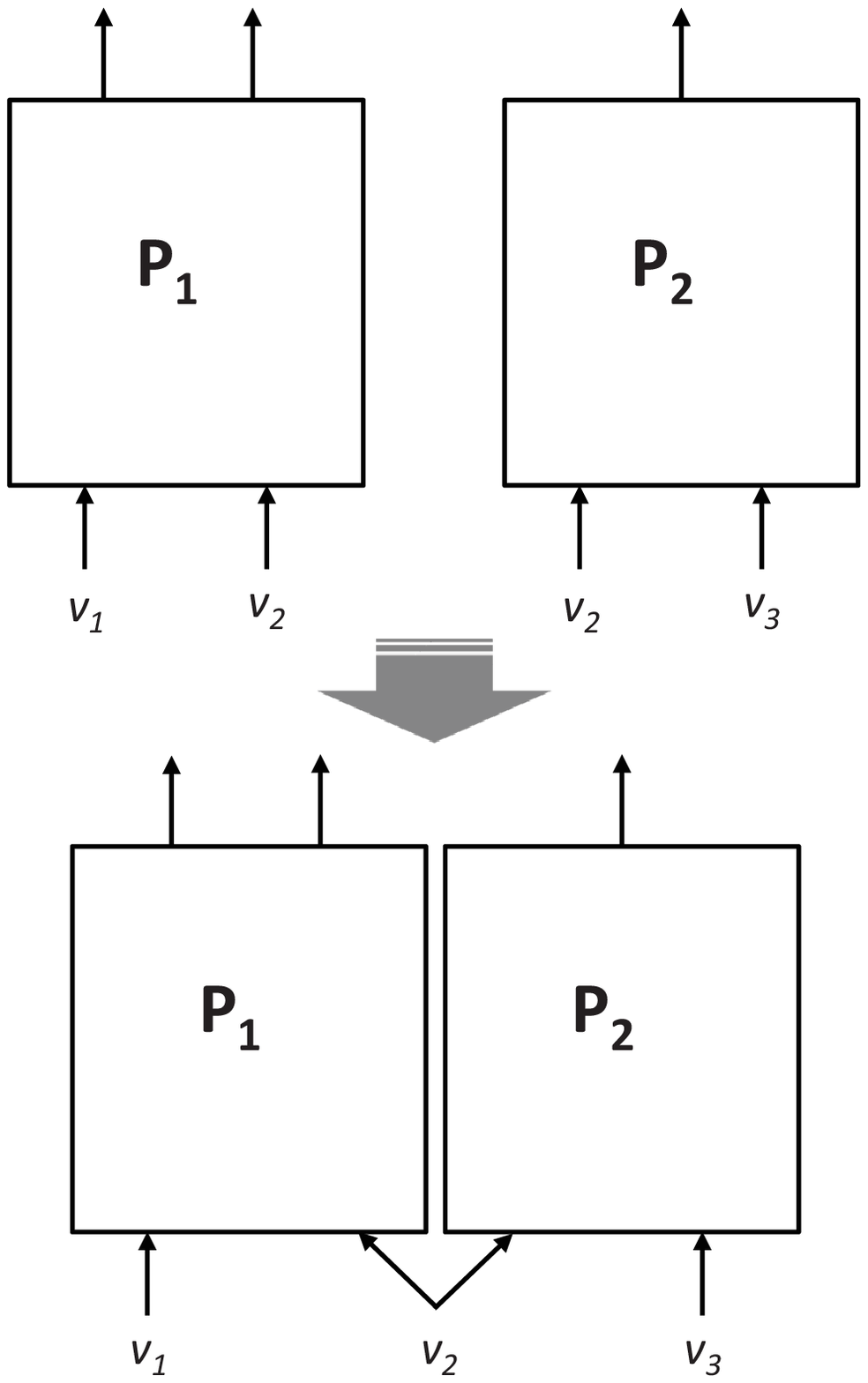}
\caption{Parallel composition of $P_1$ and $P_2$.}
 \label{fig:parcomp}
\end{center}
\vspace{-1.2cm}
\end{wrapfigure}
This section shows the compositionality property of the
$n$-threshold-probing model of leakage resilience.  We state and prove
two main results (for space the proofs are deferred to
\addappendix{Appendix}{the extended report~\cite{extendedreport}}).  

The first result concerns {\em parallel compositions}.  It shows that
given two $n$-leakage-resilient circuits $P_1$ and $P_2$ that possibly
share inputs, the composed circuit that runs $P_1$ and $P_2$ in
parallel is also $n$-leakage resilient, assuming that the randoms in
the two components are disjoint.  Fig.~\ref{fig:parcomp} shows the
diagram depicting the composition.  The second result concerns {\em
  sequential compositions}, and it is significantly harder to prove
than the first one.  The sequential composition takes an
$n$-leakage-resilient circuit $P_2$ having $y$ many (non-split)
inputs, and $n$-leakage-resilient circuits $P_{11},\dots,P_{1y}$ each
having one (non-split) output.  The composition is done by connecting
each split output of the output-decoder-free part of $P_{1i}$ to the
$i$th split input of the input-encoder-free part of $P_2$.  Clearly,
the composed circuit is IO-equivalent to the one formed by connecting
each non-split output of $P_{1i}$ to the $i$th non-split input of
$P_2$.  The sequential compositionality result states that the
composed circuit is also $n$-leakage resilient, under the assumption
that the randoms in the composed components are disjoint.
Fig.~\ref{fig:seqcomp} shows the diagram of the composition.  We state
and prove the parallel compositionality result formally in
Section~\ref{sec:parcomp}, and the sequential compositionality result
in Section~\ref{sec:seqcomp}.

We remark that, in the sequential composition, if a (non-split) secret
input, say $\svar$, is shared by some $P_{1i}$ and $P_{1j}$ for $i
\neq j$, then the disjoint randomness condition requires $\svar$ to be
encoded by two independent input encoders.  This is in contrast to the
previous
works~\cite{DBLP:conf/fse/CoronPRR13,DBLP:journals/iacr/BartheBDFG15}
that only use one input encoder per a secret input.  On the other
hand, such works require additional randoms at the site of the
composition, whereas no additional randoms are needed at the
composition site in our case as it directly connects the split outputs
of $P_{1i}$'s to the split inputs of $P_2$.

\begin{figure}
\begin{center}
\includegraphics[width=10cm]{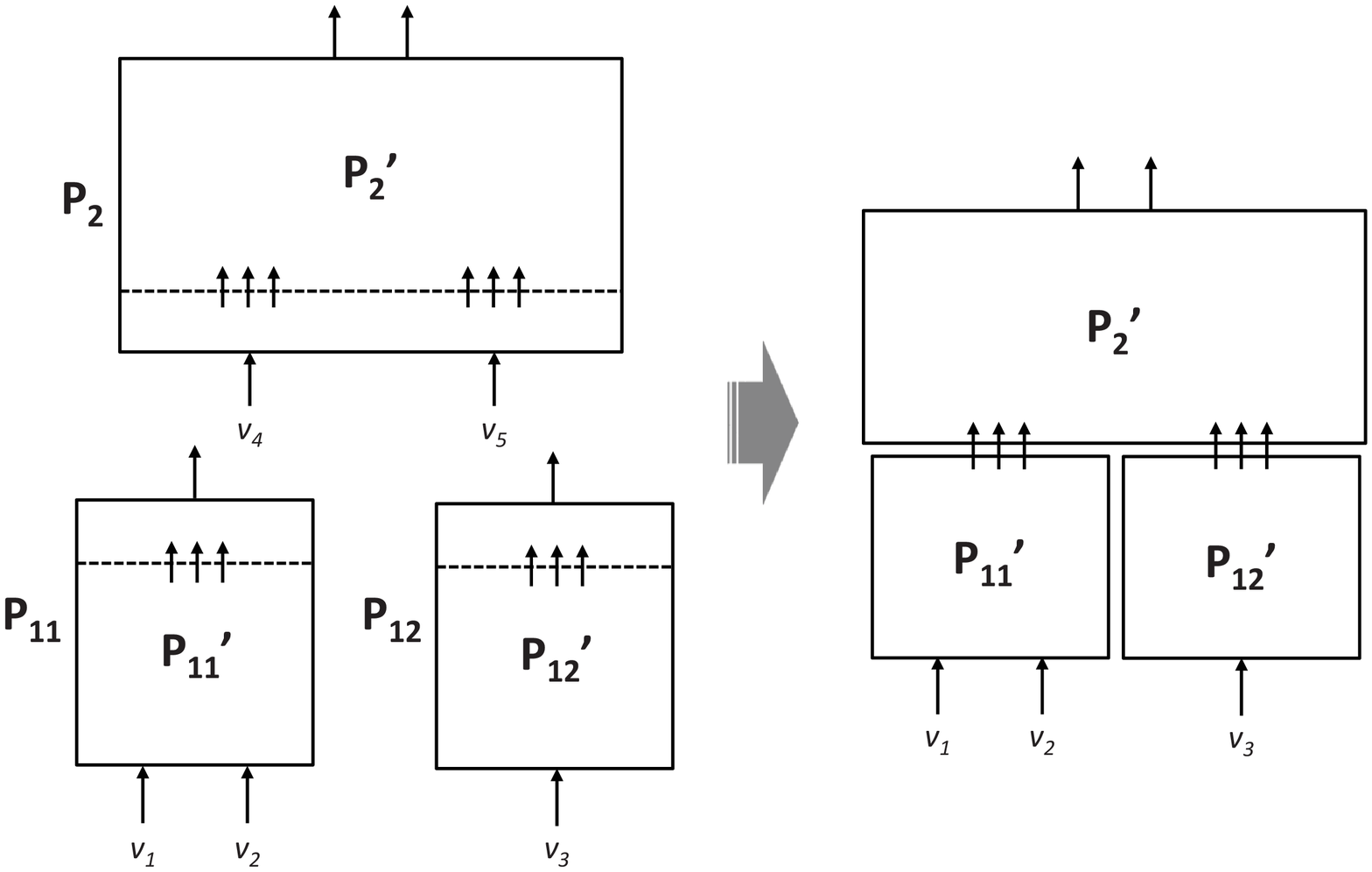}
\caption{Sequential composition of $P_{11}$, $P_{12}$, and $P_2$.
  Here, ${P_{11}}'$ (resp.~${P_{12}}'$) is the output-decoder-free part of
  $P_{11}$ (resp.~$P_{12}$), and ${P_2}'$ is the input-encoder-free part
  of $P_2$.  The composition connects the split outputs of ${P_{11}}'$
  and ${P_{12}}'$ to the split inputs of ${P_2}'$.}
 \label{fig:seqcomp}
\end{center}
\end{figure}

\subsection{Parallel Composition}
\label{sec:parcomp}

This subsection proves the parallel compositionality
result.  Let us write $\parcomp{P_1}{P_2}$ for the parallel
composition of $P_1$ and $P_2$.  We state and prove the result.
\begin{theorem}
\label{thm:parcomp}
Let $P_1$ and $P_2$ be $n$-leakage-resilient circuits having disjoint
randoms.  Then, $\parcomp{P_1}{P_2}$ is also $n$-leakage-resilient.
\end{theorem}

\begin{remark}
  While Theorem~\ref{thm:parcomp} only states that
  $\parcomp{P_1}{P_2}$ can withstand an attack that observes up to $n$
  nodes total from the composed circuit, a stronger property can
  actually be derived from the proof of the theorem.  That is, the
  proof shows that $\parcomp{P_1}{P_2}$ can withstand an attack that
  observes up to $n$ nodes from the $P_1$ part and up to $n$ nodes
  from the $P_2$ part.  (However, it is not secure against an attack
  that picks more than $n$ nodes in an arbitrary way: for example,
  picking $n+1$ nodes from one side.)
\end{remark}

\subsection{Sequential Composition}

\label{sec:seqcomp}

This subsection proves the sequential compositionality result.  As
remarked above, the result is significantly harder to prove than the
parallel compositionality result.  Let us write
$\compcirc{P_{11},\dots,P_{1y}}{P_2}$ for the sequential composition
of $P_{11},\dots,P_{1y}$ with a $y$-input circuit $P_2$.  We state and
prove the sequential compositionality result.
\begin{theorem}
\label{thm:seqmultinput}
Let $P_{11},\dots,P_{1y}$ be $n$-leakage-resilient circuits, and $P_2$
be an $y$-input $n$-leakage-resilient circuit, having disjoint
randoms.  Then, $\compcirc{P_{11},\dots,P_{1y}}{P_2}$ is
$n$-leakage-resilient.
\end{theorem}

\begin{figure}[t]
\begin{center}
\includegraphics[width=9cm]{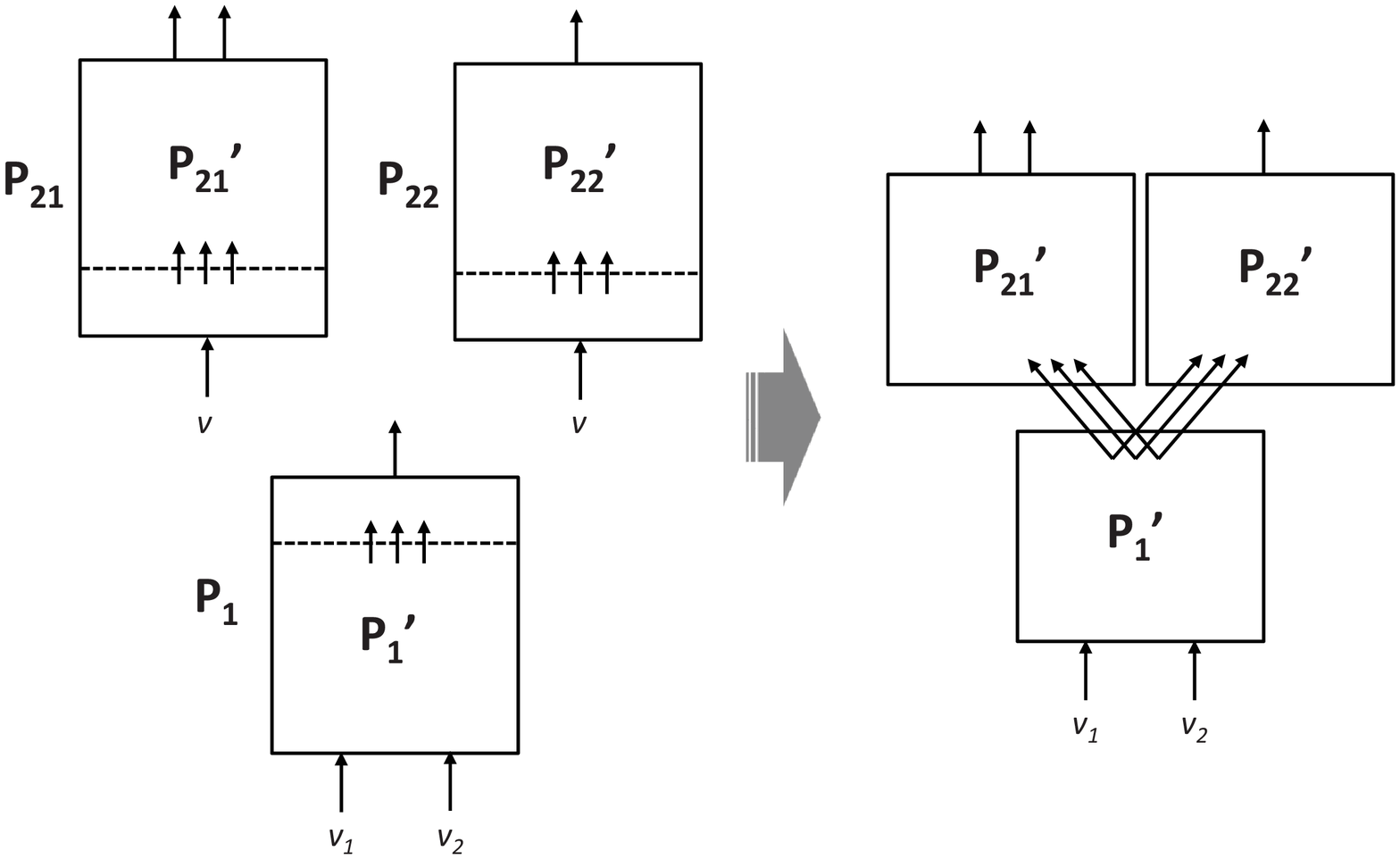}
\caption{Output-sharing sequential composition of $P_1$, $P_{21}$, and
  $P_{22}$.  Here, ${P_1}'$ is the output-decoder-free part of $P_1$,
  and ${P_{21}}'$ (resp.~${P_{22}}'$) is the input-encoder-free part
  of $P_{21}$ (resp.~$P_{22}$).  The composition connects the split
  output of ${P_1}'$ to the split inputs of ${P_{21}}'$ and ${P_{22}}'$.}
 \label{fig:outshare}
\end{center}
\end{figure}

\begin{example}
\label{ex:seqtight}
As remarked in Section~\ref{sec:comp}, the parallel
compositionality result enjoys an additional property that the
circuit is secure even under an attack that observes more than $n$
nodes in the composition as long as the observation in each component
is at most $n$.  We show that the property does not hold in the case
of sequential composition.  Indeed, it can be shown that just allowing
$n+1$ observations breaks the security even if the number of
observations made within each component is at most $n$.  

{\makeatletter
\let\par\@@par
\par\parshape0
\everypar{}\begin{wrapfigure}{r}{14em}
\vspace{-2em}
\begin{center}
\includegraphics[width=12em]{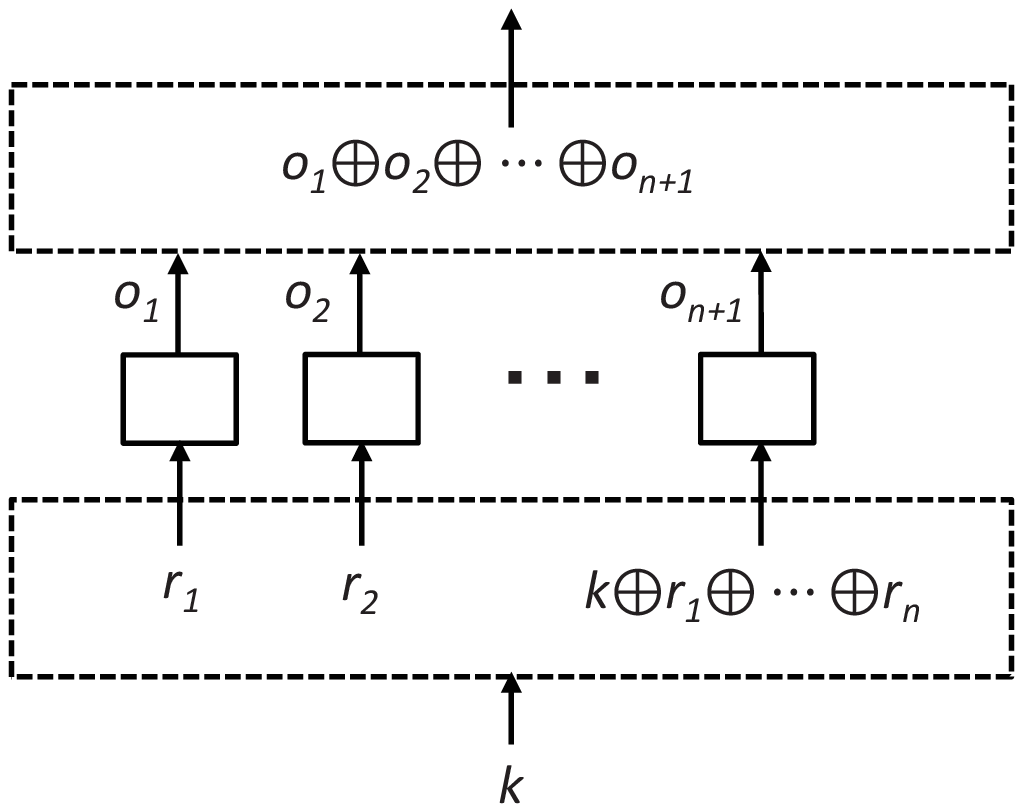}
\caption{An $n$-leakage resilient identity circuit.}
 \label{fig:id}
\vspace{-2em}
\end{center}
  \end{wrapfigure}
To see this, consider the \plusone{n}-split circuit shown in
Fig.~\ref{fig:id}.  The circuit implements the identity function,
and it is easy to see that the circuit is $n$-leakage resilient.  Let
$P_1$ and $P_2$ be copies of the circuit, and consider the composition
$\compcirc{P_1}{P_2}$.  Then, the composed circuit is not secure
against an attack that observes $m$ nodes of $P_1$ for some $1 \leq m
\leq n$, and observes $n+1 - m$ nodes of $P_2$ such that the nodes
picked on the $P_2$ side are the nodes connected to the nodes that are
not picked on the $P_1$ side.\par}
\end{example}

\begin{remark}
\label{rem:outputshare}
By a reasoning similar to the one used in the proof of
Theorem~\ref{thm:seqmultinput}, we can show the correctness of a more
parsimonious version of parallel composition theorem
(Theorem~\ref{thm:parcomp}) where given $P_1$ and $P_2$ that shares a
secret, instead of $\parcomp{P_1}{P_2}$ duplicating split shares of
the secret as in Fig.~\ref{fig:parcomp}, we make one split share
tuple to be used in the both sides of the composition.  Combining this
improved parallel composition with the sequential compositionality
result, we obtain compositionality for the case where an output of a
circuit is shared by more than one circuit in a sequential
composition.

Fig.~\ref{fig:outshare} depicts such an {\em output-sharing
  sequential composition}.  Here, $P_1$, $P_{21}$, and $P_{22}$ are
$n$-leakage-resilient circuits, and we wish to compose them by
connecting the output of $P_1$ to the input of $P_{21}$ and the input
of $P_{22}$.  By the parallel compositionality result, the parallel
composition of $P_{21}$ and $P_{22}$ that shares the same input
($\var$) is $n$-leakage-resilient.  Then, it follows that,
sequentially composing that parallelly composed circuit with $P_1$, as
depicted in the figure, is also $n$-leakage-resilient thanks to the
sequential compositionality result.
\end{remark}

\section{Compositional Synthesis Algorithm}
\label{sec:alg}

\renewcommand{\algorithmicrequire}{\textbf{Input:}}
\renewcommand{\algorithmicensure}{\textbf{Output:}}

The compositionality property gives a way for a compositional approach
to synthesizing leakage-resilient circuits.  Algorithm~\ref{alg:main}
shows the overview of the synthesis algorithm.  Given a random-free
circuit as an input, the algorithm synthesizes an IO-equivalent
$n$-leakage-resilient circuit.  It first invokes the $\textsc{Decomp}$
operation to choose a suitable decomposition of the given circuit into
some number of sub-circuits.  Then, it invokes $\textsc{MonoSynth}$
on each sub-circuit $P_i$ to synthesize an $n$-leakage resilient
circuit ${P_i}'$ that is IO-equivalent to $P_i$.  Finally, it returns
the composition of the obtained circuits as the synthesized
$n$-leakage resilient version of the original.

\begin{algorithm}
\caption{The Compositional Synthesis Algorithm}         
\label{alg:main}                          
\begin{algorithmic}[1]                 
\Require Random-free Circuit $P$
\Ensure IO-equivalent $n$-leakage-resilient circuit 
\vspace{0.2em}
\State $P_1,\dots,P_m$ $\assign$ $\textsc{Decomp}(P)$
\For{\textbf{each} $P_i \in \{P_1,\dots,P_m\}$}
\State ${P_i}'$ $\assign$ $\textsc{MonoSynth}(P_i)$
\EndFor
\State \Return $\textsc{Comp}({P_1}',\dots,{P_m}')$
\end{algorithmic}
\end{algorithm}

$\textsc{Comp}$ is the composition operation, and it composes the
given $n$-leakage-resilient circuits in the manner described in
Section~\ref{sec:comp}. $\textsc{MonoSynth}$ is a constraint-based
``monolithic'' synthesis algorithm that synthesizes an
$n$-leakage-resilient circuit that is IO-equivalent to the given
circuit without further decomposition.  We describe
$\textsc{MonoSynth}$ in Section~\ref{sec:mono}, and describe the
decomposition operation $\textsc{Decomp}$ in Section~\ref{sec:decomp}.

The algorithm optimizes the synthesized circuits in the following
ways.  First, as described in Section~\ref{sec:mono}, the monolithic
synthesis looks for tree-shaped circuits of the shortest tree height.
Secondly, as described in Section~\ref{sec:decomp}, the decomposition
and composition is done in a way to avoid unnecessarily making the
non-secret-dependent parts leakage resilient, and also to re-use the
synthesis results for shared sub-circuits whenever allowed by the
compositionality properties.

\begin{remark}
\label{rem:comp}
The compositional algorithm composes the $n$-leakage-resilient
versions of the sub-circuits.  Note that the compositionality property
states that the result will be $n$-leakage-resilient after the
composition, regardless of how the sub-circuits are synthesized as
long as they are also $n$-leakage-resilient and have disjoint randoms.
Thus, in principle, any method to synthesize the $n$-leakage-resilient
versions of the sub-circuits may be used in place of
$\textsc{MonoSynth}$.  For instance, a possible alternative is to use
a database of $n$-leakage-resilient versions of commonly-used circuits
(e.g., obtained via the construction
of~\cite{DBLP:conf/crypto/IshaiSW03,DBLP:conf/fse/CoronPRR13}).
\end{remark}

\subsection{Constraint-Based Monolithic Synthesis}
\label{sec:mono}

The monolithic synthesis algorithm is based on and extends the
constraint-based approach proposed by Eldib and
Wang~\cite{DBLP:conf/cav/EldibW14}.  The algorithm synthesizes an
$n$-leakage-resilient circuit that is IO-equivalent to the given
circuit.  The algorithm requires the given circuit to have only one
output.  Therefore, the overall algorithm to decomposes the whole
circuit into such sub-circuits before passing them to
$\textsc{MonoSynth}$.

We formalize the algorithm as quantified first-order logic constraint
solving.  Let $P$ be the random-free circuit
given as the input.  We prepare a quantifier-free formula
$\phip(\ve{\nd},\ve{\pvar},\ve{\svar},\ovar)$ on the free variables
$\ve{\nd}$, $\ve{\pvar}$, $\ve{\svar}$, $\ovar$ that encodes the
input-output behavior of $P$.  Formally, $\exists
\ve{\nd}.\phip(\ve{\nd},\ve{\pvar},\ve{\svar},\ovar)$ is true iff $P$
outputs $\ovar$ given public inputs $\ve{\pvar}$ and secret inputs
$\ve{\svar}$.  The variables $\ve{\nd}$ are used for encoding the
shared sub-circuits within $P$ (i.e., gates of fan-out $>1$).  For
example, for $P$ that outputs $\svar \wedge \pvar$,
$\phip(\pvar,\svar,\ovar) \equiv \ovar = \svar \wedge \pvar$.

\begin{figure}[t]
\begin{center}
\begin{tabular}[t]{cc}
\begin{minipage}[c]{12em}
\includegraphics[width=10em]{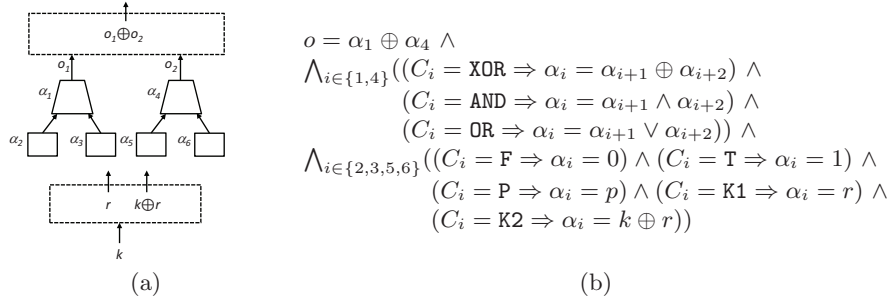}
\end{minipage}
&
\(
\begin{array}{l}
\ovar = \nd_1 \oplus \nd_4 \ \wedge \\
\bigwedge_{i \in \aset{1,4}} ((\cvar_i = \texttt{XOR} \Rightarrow \nd_i = \nd_{i+1} \oplus \nd_{i+2}) \ \wedge \\
\hspace{4em} (\cvar_i = \texttt{AND} \Rightarrow \nd_i = \nd_{i+1} \wedge \nd_{i+2}) \ \wedge \\
\hspace{4em}  (\cvar_i = \texttt{OR} \Rightarrow \nd_i = \nd_{i+1} \vee \nd_{i+2})) \ \wedge \\
\bigwedge_{i \in \aset{2,3,5,6}} ((\cvar_i = \texttt{F} \Rightarrow \nd_i = 0) \wedge
(\cvar_i = \texttt{T} \Rightarrow \nd_i = 1) \ \wedge\\
\hspace{5.2em}  (\cvar_i = \texttt{P} \Rightarrow \nd_i = \pvar) \wedge
(\cvar_i = \texttt{K1} \Rightarrow \nd_i = \rvar) \ \wedge\\
\hspace{5.2em} (\cvar_i = \texttt{K2} \Rightarrow \nd_i = \svar \oplus \rvar))
\end{array}
\)
\\
\textrm{(a)} & \textrm{(b)}
\end{tabular}
\caption{(a) Skeleton circuit.  (b) $\phiskel{2}$ for the skeleton circuit.}
 \label{fig:exskel}
\end{center}
\end{figure}

Adopting the approach of \cite{DBLP:conf/cav/EldibW14}, our monolithic
algorithm works by preparing {\em skeleton circuits} of increasing
size, and searching for an instance of the skeleton that is
$n$-leakage-resilient and IO-equivalent to $P$.  By starting from a
small skeleton, the algorithm biases toward finding an optimized
leakage resilient circuit.  Formally, a skeleton circuit $\skel{\ell}$
is a tree-shaped circuit (i.e., circuit with all non-input gates
having fan-out of $1$) of height $\ell$ whose gates have undetermined
functionality except for the parts that implement the input encoder
and the output decoder.  For example, Fig.~\ref{fig:exskel} (a) shows
the \plusone{1}-split skeleton circuit of height $2$ with one secret
input.

We prepare a quantifier-free {\em skeleton formula} $\phiskel{\ell}$
that expresses the skeleton circuit.  Formally,
$\phiskel{\ell}(\ve{\cvar},
\ve{\nd},\ve{\pvar},\ve{\svar},\ve{\rvar},\ovar)$ is true iff $P'$
outputs $\ovar$ with the valuations of $\nodes{P'}$ having the values
$\ve{\nd}$ given public inputs $\ve{\pvar}$, secret inputs
$\ve{\svar}$, and random inputs $\ve{\rvar}$ where $P'$ is
$\skel{\ell}$ with its nodes' functionality determined by
$\ve{\cvar}$.\footnote{Technically, the number of randoms (i.e.,
  $\ve{\rvar}$) can also be left undetermined in the skeleton.  Here,
  for simplicity, we assume that the number of randoms is determined
  by the factors such as $n$, the skeleton height, and the number of
  secret inputs.}  We call $\ve{\cvar}$ the {\em control variables},
and write $\skel{\ell}(\ve{\cvar})$ for the instance of $\skel{\ell}$
determined by $\ve{\cvar}$.  For example, Fig.~\ref{fig:exskel} (b)
shows $\phiskel{\ell}$ for the skeleton circuit from
Fig.~\ref{fig:exskel} (a), when there is one public input and no
randoms besides the one from the input encoder.

The synthesis is now reduced to the problem of finding an assignment
to $\ve{\cvar}$ that satisfies the constraint $\phiioeq(\ve{\cvar})
\wedge \philr(\ve{\cvar})$ where $\phiioeq(\ve{\cvar})$ expresses that
$\skel{\ell}(\ve{\cvar})$ is IO-equivalent to $P$, and
$\philr(\ve{\cvar})$ expresses that $\skel{\ell}(\ve{\cvar})$ is
$n$-leakage-resilient.  As we shall show next, the constraints
faithfully encode IO-equivalence and leakage-resilience according to
the definitions from Section~\ref{sec:prelim}.

$\phiioeq(\ve{\cvar})$ is the formula below.
\[
\begin{array}{l}
\forall \ve{\nd},\ve{\nd'},\ve{\pvar},\ve{\svar},\ve{\rvar},\ovar,\ovar'.
\ \phip(\ve{\nd},\ve{\pvar},\ve{\svar},\ovar) \wedge 
\phiskel{\ell}(\ve{\cvar},\ve{\nd'},\ve{\pvar},\ve{\svar},\ve{\rvar},\ovar')
\Rightarrow \ovar = \ovar'
\end{array}
\]
It is easy to see that $\phiioeq$ correctly expresses IO equivalence.

The definition of $\philr$ is more involved, and we begin by
introducing a useful shorthand notation.  Let $m$ be the number of
(observable) nodes in $\skel{\ell}$.  Without loss of generality, we
assume that $m \geq n$.\texcomment{Because otherwise, circuit cannot
  be made secure.  Attacker can just see everything out of the input
  encoder or all inputs to the output decoder.  Also, constraints to
  handle that case is easy because there is only one possible
  selection for the attacker - all nodes.}  Let $\ve{\sigma} \in
\aset{1,\dots,m}^*$ be a sequence such that $|\ve{\sigma}| \leq n$,
and $\ve{\nd} = \nd_1,\dots,\nd_m$ be a length $m$ sequence of
variables.  We write $\nodesel{\ve{\nd}}{\ve{\sigma}}$ for the
sequence of variables $\ndb_1,\dots \ndb_{|\ve{\sigma}|}$ such that
$\ndb_i = \nd_{\sigma_i}$ for each $i \in
\aset{1,\dots,|\ve{\sigma}|}$.  Intuitively, $\ve{\sigma}$ represents
a selection of nodes observed by the adversary.  For example, let
$\ve{\alpha}=\alpha_1,\alpha_2,\alpha_3$ and $\ve{\sigma} = 1,3$, then
$\nodesel{\ve{\nd}}{\ve{\sigma}} = \nd_1,\nd_3$.

Let $\rspace = \aset{0,1}^{|\ve{\rvar}|}$.  Then, $\philr(\ve{\cvar})$
is the formula below.
\[
\begin{array}{l}
\forall \ve{\sigma}.
\forall \ve{\ndb},\ve{\nd}',\ve{\pvar},\ve{\svar},\ve{\svar}',\ve{\ovar}.\\
\hspace{1.5em} \bigwedge_{\ve{b} \in \rspace} \phi[\ve{\nd_b}/\ve{\nd}][\ovar_{\ve{b}}/\ovar] \wedge 
\phi[\ve{\nd'_b}/\ve{\nd}][\ovar'_{\ve{b}}/\ovar][\ve{\svar'}/\ve{\svar}]\\
 \hspace{3em} \Rightarrow \ \sum_{\ve{b} \in \rspace} (\nodesel{\ve{\nd_b}}{\ve{\sigma}} = \ve{\ndb}) = 
\sum_{\ve{b} \in \rspace} (\nodesel{\ve{\nd_b}'}{\ve{\sigma}} = \ve{\ndb})
\end{array}
\]
where $\ve{\nd}'$ is a sequence comprising distinct variables
$\ve{\nd_b}$ and $\ve{\nd'_b}$ such that $|\ve{\nd_b}| = |\ve{\nd'_b}|
= m$ for each $\ve{b} \in \rspace$, $\ve{\ovar}$ is a sequence
comprising distinct variables $\ovar_{\ve{b}}$ and $\ovar'_{\ve{b}}$
for each $\ve{b} \in \rspace$, and $\phi$ is the formula
$\phiskel{\ell}(\ve{\cvar},\ve{\nd},\ve{\pvar},\ve{\svar},\ve{b},\ovar)$.
While $\philr(\ve{\cvar})$ is not strictly a first-order logic
formula, it can be converted to the form by expanding the finitely many
possible choices of $\ve{\sigma}$.

\begin{algorithm}
\caption{\textsc{MonoSynth}}         
\label{alg:mono}                          
\begin{algorithmic}[1]                 
\Require Random-free Circuit $P$
\Ensure IO-equivalent $n$-leakage resilient circuit 
\vspace{0.2em}
\State $\ell$ $\assign$ $\texttt{INIT\_HEIGHT}$
\Loop
\State $\testset$ $\assign$ $\emptyset$; $\nodecount$ $\assign$ $n$
\Loop
\State \textbf{match} \textsc{FindCand}($\testset$,$\gamma$) \textbf{with}
\State \hspace{1em} $\nosol$ $\rightarrow$ $\ell$\texttt{++}; \textbf{break}
\State \hspace{0.5em} $\mid$ $\sol(\ve{b_\cvar})$ $\rightarrow$ 
\textbf{match} \textsc{CheckCand}($\ve{b_\cvar}$) \textbf{with}
\State \hspace{7.5em} $\success$ $\rightarrow$ \Return $\skel{\ell}(\ve{b_\cvar})$
\State \hspace{7em} $\mid$ $\cex(\testset',\nodecount')$ $\rightarrow$ $\testset$ $\assign$ $\testset \cup \testset'$; $\nodecount$ $\assign$ $\textit{max}(\nodecount,\nodecount')$
\EndLoop
\EndLoop
\end{algorithmic}
\end{algorithm}

Because the domains of the quantified variables in $\phiioeq$ and
$\philr$ are finite, one approach to solving the constraint may be first
eliminating the quantifiers eagerly and then solving for $\ve{\cvar}$
that satisfies the resulting constraint.  However, the approach is
clearly impractical due to the extreme size of the resulting
constraint.  Instead, adopting the idea from
\cite{DBLP:conf/cav/EldibW14}, we solve the constraint by lazily
instantiating the quantifiers.  The main idea is to divide the
constraint solving process in two phases: the {\em candidate finding}
phase that infers a candidate solution for $\ve{\cvar}$, and the {\em
  candidate checking} phase that checks whether the candidate is
actually a solution.  We run the phases iteratively in a CEGAR style
until convergence.  Algorithm~\ref{alg:mono} shows the overview of the
process.  We describe the details of the algorithm below.

\vspace{0.5em}
\noindent{\bf Candidate Checking.}
The candidate checking phase is straightforward.  Note that, after
expanding the choices of $\ve{\sigma}$ in $\philr$,
$\phiioeq(\ve{\cvar}) \wedge \philr(\ve{\cvar})$ only has outer-most
$\forall$ quantifiers.  Therefore, given a concrete assignment to
$\ve{\cvar}$, $\ve{b_\cvar}$, $\textsc{CheckCand}$ directly solves the
the constraint by using an SMT solver.\footnote{Also, because
  $\ve{\cvar}$ are the only shared variables in $\phiioeq$ and
  $\philr$, the two may be checked independently after instantiating
  $\ve{\cvar}$.}  (However, naively expanding $\ve{\sigma}$ can be
costly when $n > 1$, and we show a modification that alleviates the
cost in the later part of the subsection.)

\vspace{0.5em}
\noindent{\bf Candidate Finding.}
We describe the candidate finding process $\textsc{FindCand}$.  To
find a likely candidate, we adopt the idea from
\cite{DBLP:conf/cav/EldibW14} and prepare a {\em test set} that is
updated via the CEGAR iterations.  In \cite{DBLP:conf/cav/EldibW14}, a
test set, $\testset$, is a pair of sets $\testset_\pvar$ and
$\testset_\svar$ where $\testset_\pvar$ (resp.~$\testset_\svar$)
contains finitely many concrete valuations of $\ve{\pvar}$
(resp.~$\ve{\svar}$).  Having such a test set, we can rewrite the
constraint so that the public inputs and secret inputs are restricted
to those from the test set.  That is, $\phiioeq(\ve{\cvar})$ is
rewritten to be the formula below.
\[
\begin{array}{l}
\bigwedge_{(\ve{b_\pvar},\ve{b_\svar}) \in \testset_{\pvar}\times\testset_{\svar}}\\
\hspace{0.5em} \phip(\ve{\nd_{\tst{}}},\ve{b_\pvar},\ve{b_\svar},\ovar_{\tst{}})  \wedge \phiskel{\ell}(\ve{\cvar},\ve{\nd'_{\tst{}}},\ve{b_\pvar},\ve{b_\svar},\ve{\rvar_{\tst{}}},\ovar'_{\tst{}}) 
\wedge \ovar_{\tst{}} = \ovar'_{\tst{}}
\end{array}
\]
And, $\philr(\ve{\cvar})$ becomes the formula below.
\[
\begin{array}{l}
\forall \ve{\sigma}. \forall \ve{\ndb}.\bigwedge_{(\ve{b_\pvar},\ve{b_\svar},\ve{b_\svar}') \in \testset_\pvar\times\testset_\svar\times\testset_\svar}\bigwedge_{\ve{b} \in \rspace} \\
\hspace{1em} \phi[\ve{\nd}_{\tst{}\ve{b}}/\ve{\nd}][\ovar_{\tst{}\ve{b}}/\ovar][\ve{b_\svar}/\ve{\svar}] \wedge
\phi[\ve{\nd}'_{\tstprime{}\ve{b}}/\ve{\nd}][\ovar'_{\tstprime{}\ve{b}}/\ovar][\ve{b_\svar'}/\ve{\svar}] \\
\hspace{1em} \wedge \ \sum_{\ve{b} \in \rspace} (\nodesel{\ve{\nd_{\tst{}\ve{b}}}}{\ve{\sigma}} = \ve{\ndb}) = 
\sum_{\ve{b} \in \rspace} (\nodesel{\ve{\nd'_{\tstprime{}\ve{b}}}}{\ve{\sigma}} = \ve{\ndb})
\end{array}
\]
where $\phi$ is the formula
$\phiskel{\ell}(\ve{\cvar},\ve{\nd},\ve{b_\pvar},\ve{\svar},\ve{b},\ovar)$.
We remark that, because fixing the inputs to concrete values also
fixes the valuations of some other variables (e.g., fixing
$\ve{\pvar}$ and $\ve{\svar}$ also fixes $\ve{\nd}$ and $\ovar$ in
$\phip(\ve{\nd},\ve{\pvar},\ve{\svar},\ovar)$), the constraint
structure is modified to remove the quantifications on such variables.

At this point, the approach of \cite{DBLP:conf/cav/EldibW14} can be
formalized as the following process: it eagerly instantiates the
possible choices of $\ve{\sigma}$ and $\ve{\ndb}$ to reduce the
constraint to a quantifier-free form, and looks for a satisfying
assignment to the resulting constraint.  This is a sensible approach
when $n$ is $1$ because, in that case, the number of possible choices
of $\ve{\sigma}$ is linear in the size of the skeleton (i.e., is $m$)
and the possible valuations of $\ve{\ndb}$ are simply $\aset{0,1}$.
Unfortunately, the number of possible choices of $\ve{\sigma}$ grows
exponentially in $n$, and so does that of the possible valuations of
$\ve{\ndb}$.\footnote{Therefore, the complexity of the method is at
  least double exponential in $n$ assuming that the complexity of the
  constraint solving process is at least exponential in the size of
  the given formula.}  We remark that this is expected because
$\ve{\sigma}$ represents the adversary's node selection choice, and
$\ve{\ndb}$ represents the valuation of the selected nodes.  Indeed,
in our experience with a prototype implementation, this method
fails even on quite small sub-circuits when $n$ is even $2$.

Therefore, we make the following improvements to the base algorithm.
\begin{itemize}
\item[(1)] We restrict the node selection to root-most $\nodecount$
  nodes where $\nodecount$ starts at $n$ and is incremented
  via the CEGAR loop.
\item[(2)] We include node valuations in the test set.
\item[(3)] We use dependency analysis to reduce irrelevant
node selections from the constraint in the candidate checking phase.
\end{itemize}
The rationale for prioritizing the root-most nodes in (1) is that, in
a tree-shaped circuit, nodes closer to the root are more likely to be
dependent on the secret and therefore are expected to be better
targets for the adversary.  The number of root-most nodes to select,
$\nodecount$, is incremented as needed by a counterexample analysis
(cf.~lines 7-9 of Algorithm~\ref{alg:mono}).  The test set generation
for node valuations described in item (2) is done in much the same way
as that for public inputs and secret inputs.  We describe the test set
generation process in more detail in {\bf Test Set Generation}.  With
the modifications (1) and (2), the leakage-resilience constraint to be
solved in the candidate finding phase is now the following formula.
\[
\begin{array}{l}
\forall \ve{\sigma}\smallcolon\nodecount. \bigwedge_{(\ve{b_\pvar},\ve{b_\svar},\ve{b_\svar}',\ve{b_\ndb}) \in \testset_\pvar\times\testset_\svar\times\testset_\svar\times\testset_\ndb} \bigwedge_{\ve{b} \in \rspace} \\
\hspace{1em} \phi[\ve{\nd}_{\tst{}\ve{b}}/\ve{\nd}][\ovar_{\tst{}\ve{b}}/\ovar][\ve{b_\svar}/\ve{\svar}] \wedge
\phi[\ve{\nd}'_{\tstprime{}\ve{b}}/\ve{\nd}][\ovar'_{\tstprime{}\ve{b}}/\ovar][\ve{b_\svar'}/\ve{\svar}] \\
\hspace{1em} \wedge \ \sum_{\ve{b} \in \rspace} (\nodesel{\ve{\nd_{\tst{}\ve{b}}}}{\ve{\sigma}} = \ve{b_\ndb}) = 
\sum_{\ve{b} \in \rspace} (\nodesel{\ve{\nd'_{\tstprime{}\ve{b}}}}{\ve{\sigma}} = \ve{b_\ndb})
\end{array}
\]
where $\ve{\sigma}\smallcolon\nodecount$ restricts $\ve{\sigma}$
to the root most $\nodecount$ indexes, $\testset_\ndb$ is the
set of test set elements for node valuations, and $\phi$ is the
formula
$\phiskel{\ell}(\ve{\cvar},\ve{\nd},\ve{b_\pvar},\ve{\svar},\ve{b},\ovar)$.

Unlike (1) and (2), the modification (3) applies to the candidate
checking phase.  To see the benefit of this modification, note that,
even in the candidate checking phase, checking the leakage-resilience
condition $\philr(\ve{b_\cvar})$ can be quite expensive because it
involves expanding exponentially many possible choices of node
selections.  To mitigate the cost, we take advantage of the fact that
the candidate circuit is fixed in the candidate checking phase, and do
a simple dependency analysis on the candidate circuit to reduce
irrelevant node-selection choices.  We describe the modification in
more detail.  Let $P'$ be the candidate circuit.  For each node of
$P'$, we collect the reachable leafs from the node to obtain the
over-approximate set of inputs on which the node may depend.  For a
node $\nd$ of $P'$, let $\depins{\nd}$ be the obtained set of
dependent inputs for $\nd$. Then, any selection of nodes $\ve{\nd}$
such that $\bigcup_{\nd \in \aset{\ve{\nd}}} \depins{\nd}$ does not
contain all \plusone{n}-split shares of some secret is an irrelevant
selection and can be removed from the constraint.  (Here, we use the
symbols $\nd$ for node labels as in Section~\ref{sec:prelim}, and not
as node-valuation variables in a constraint.)

\vspace{0.5em}
\noindent{\bf Test Set Generation.}
Recall that our algorithm maintains three kinds of test sets,
$\testset_\pvar$ for public inputs, $\testset_\svar$ for secret
inputs, and $\testset_\ndb$ for node valuations.  As shown in lines
7-9 of Algorithm~\ref{alg:mono}, we obtain new test set elements from
candidate check failures (here, by abuse of notation, we write
$\testset \cup \testset'$ for the component-wise union).  We describe
the process in more detail.  In $\textsc{CheckCand}$, we convert the
constraint $\phiioeq(\ve{b_\cvar}) \wedge \philr(\ve{b_\cvar})$ to a
quantifier free formula $\Phi$ by expanding the selection choices and
removing the universal quantifiers.  Then, we use an SMT solver to
check the satisfiability of $\neg \Phi$ and return $\success$ if it is
unsatisfiable.  Otherwise, the SMT solver returns a satisfying
assignment of $\neg \Phi$, and we return the values assigned to
variables corresponding to public inputs, secret inputs and node
valuations as the new elements for the respective test sets.  The
number of root-most nodes to select is also raised here by taking the
maximum of the root-most nodes observed in the satisfying assignment,
$\nodecount'$, with the current $\nodecount$.

\subsection{Choosing Decomposition}
\label{sec:decomp}

This subsection describes the decomposition procedure
$\textsc{Decomp}$.  Thanks to the generality of the compositionality
results, in principle, we can decompose the given circuit into
arbitrarily small sub-circuits (i.e., down to individual gates).
However, choosing a too fine-grained decomposition may lead to a
sub-optimal result.\footnote{One may make the analogy to compiler
  optimization.  Such a decomposition strategy is
  analogous to optimizing each instruction individually.}

To this end, we have implemented the following decomposition strategy.
First, we run a dependency analysis, similar to the one used in the
constraint-based monolithic synthesis (cf.~Section \ref{sec:mono}).
The analysis result is used to identify the parts of the given circuit
that do not depend on any of the secrets.  We factor out such {\em
  public-only} sub-circuits from the rest so that they will not be
subject to the leakage-resilience transformation.

Next, we look for sub-circuits that are used at multiple locations
(i.e., whose roots have fan-out $>1$), and prioritize them to be
synthesized separately and composed at their use sites.  Besides the
saving in the synthesis effort, the approach can lead a smaller
synthesis result when the shared sub-circuit is used in contexts that
lead to different outputs (cf~Remark~\ref{rem:outputshare}).  Finally,
as a general strategy, we apply parallel composition at the root so
that we synthesize separately for each output given a multi-output
circuit.  And, we set a bound on the maximum size of the circuits that
will be synthesized monolithically, and decompose systematically based
on the bound.  As discussed in Section~\ref{sec:impl}, in the
prototype implementation, we use an ``adaptive'' version of the latter
decomposition process by adjusting the bound on-the-fly and also opts
for a pre-made circuit under certain conditions.
\begin{example}
  Let us apply the compositional synthesis algorithm to the
  circuit from Fig.~\ref{fig:exorig}, for the case $n=2$.  Note that
  the circuit has no non-trivial public-only sub-circuits or have
  non-inputs gates with fan-out greater than $1$.

  First, we apply the parallel compositionality result so that the
  circuit is decomposed to two parts: the left tree that computes
  $\pvar_1 \oplus \svar_1 \oplus \svar_2$ and the right tree that
  computes $\svar_2 \oplus \pvar_2$.  The right tree cannot be
  decomposed further, and we apply $\textsc{MonoSynth}$ to transform
  it to a leakage-resilient form.  A possible synthesis result of this
  is the right sub-circuit shown in Fig.~\ref{fig:ex} (i.e., the
  sub-circuit whose observable part outputs the split output
  $\ovar_4,\ovar_5,\ovar_6$).

  For the left tree, if the monolithic-synthesis size bound is set to
  synthesize circuits of height $2$, we apply $\textsc{MonoSynth}$
  directly to the tree.  Alternatively, with a lower bound set, we
  further decompose the left tree to a lower part that computes
  $\pvar_1 \oplus \svar_1$ and $\pvar_2$ (identity function) and an
  upper part that computes $\svar \oplus \pvar_2$ where the output of
  the lower part is to be connected to the ``place-holder'' input
  $\svar$.  Following the either strategy, we may obtain the the left
  sub-circuit of Fig.~\ref{fig:ex} as a possible result.  And, the
  final synthesis result after composing the left and right synthesis
  results is the whole circuit of Fig.~\ref{fig:ex}.
\end{example}

\section{Implementation and Experiments}

\label{sec:impl}

\begin{table*}[t]
\scriptsize
\begin{center}
\begin{tabular}{|l|c|c|c|c|c|c|c|c|c|l|c|c|c|c|c|c|c|c|c|}
\hline
\multirow{2}{*}{\legname} &
\multicolumn{3}{|c|}{$n=2$} & \multicolumn{3}{|c|}{$n=3$} & \multicolumn{3}{|c|}{$n=4$} & 
\multirow{2}{*}{\legname} & 
\multicolumn{3}{|c|}{$n=2$} & \multicolumn{3}{|c|}{$n=3$} & \multicolumn{3}{|c|}{$n=4$} 
\\
\cline{2-10}\cline{12-20}
& \legtime & \legsize & \legnumrands & \legtime & \legsize & \legnumrands & \legtime & \legsize & \legnumrands & 
& \legtime & \legsize & \legnumrands & \legtime & \legsize & \legnumrands & \legtime & \legsize & \legnumrands 
\\
\hline
P1 & 16.2s & 123 & 32 & 49.4s & 138 & 48 & 52.7s & 160 & 64 &
P10 & \multicolumn{3}{c|}{T/O} & \multicolumn{3}{c|}{M/O} & \multicolumn{3}{c|}{M/O} \\
\hline
P2 & 10.3s & 64 & 16 & 1m13s & 76 & 24 & 4m3s & 88 & 32 &
P11 & \multicolumn{3}{c|}{T/O} & \multicolumn{3}{c|}{T/O} & \multicolumn{3}{c|}{M/O} \\
\hline
P3 & \multicolumn{3}{c|}{M/O} & \multicolumn{3}{c|}{M/O} & \multicolumn{3}{c|}{M/O} &
P12 & \multicolumn{3}{c|}{T/O} & \multicolumn{3}{c|}{T/O} & \multicolumn{3}{c|}{T/O} \\
\hline
P4 & \multicolumn{3}{c|}{M/O} & \multicolumn{3}{c|}{M/O} & \multicolumn{3}{c|}{M/O} &
P13 & \multicolumn{3}{c|}{T/O} & \multicolumn{3}{c|}{T/O} & \multicolumn{3}{c|}{T/O} \\
\hline
P5 & \multicolumn{3}{c|}{M/O} & \multicolumn{3}{c|}{M/O} & \multicolumn{3}{c|}{M/O} &
P14 & \multicolumn{3}{c|}{T/O} & \multicolumn{3}{c|}{T/O} & \multicolumn{3}{c|}{M/O} \\
\hline
P6 & \multicolumn{3}{c|}{M/O} & \multicolumn{3}{c|}{M/O} & \multicolumn{3}{c|}{M/O} &
P15 & \multicolumn{3}{c|}{T/O} & \multicolumn{3}{c|}{T/O} & \multicolumn{3}{c|}{M/O} \\
\hline
P7 & \multicolumn{3}{c|}{M/O} & \multicolumn{3}{c|}{M/O} & \multicolumn{3}{c|}{M/O} &
P16 & \multicolumn{3}{c|}{T/O} & \multicolumn{3}{c|}{T/O} & \multicolumn{3}{c|}{T/O} \\
\hline
P8 & \multicolumn{3}{c|}{M/O} & \multicolumn{3}{c|}{M/O} & \multicolumn{3}{c|}{M/O} &
P17 & \multicolumn{3}{c|}{T/O} & \multicolumn{3}{c|}{T/O} & \multicolumn{3}{c|}{T/O} \\
\hline
P9 & \multicolumn{3}{c|}{T/O} & \multicolumn{3}{c|}{T/O} & \multicolumn{3}{c|}{T/O} &
P18 & \multicolumn{3}{c|}{T/O} & \multicolumn{3}{c|}{T/O} & \multicolumn{3}{c|}{T/O} \\
\hline
\end{tabular}

\end{center}
\caption{Experiment Results: Monolithic Only.}
\label{tbl:expmono}
\end{table*}

\begin{table*}[t]
\scriptsize
\begin{center}
\begin{tabular}{|l|c|c|c|c|c|c|c|c|c|c|c|c|}
\hline
\multirow{2}{*}{\legname} &
\multicolumn{4}{|c|}{$n=2$} & \multicolumn{4}{|c|}{$n=3$} & \multicolumn{4}{|c|}{$n=4$} \\
\cline{2-13}
& \legtime & \legmaxtpc & \legsize & \legnumrands & \legtime & \legmaxtpc & \legsize & \legnumrands & \legtime & \legmaxtpc & \legsize & \legnumrands
\\
\hline
P1 & 18.7s & 2.7s & 125 & 32 & 30.2s & 4.5s & 143 & 48 & 1m1s & 33.3s & 160 & 64 \\
\hline
P2 & 11.8s & 3.6s & 65 & 16 & 1m7s & 17.3s & 74 & 24 & 2m56s & 1m21s & 88 & 32 \\
\hline
P3 & 2.7s & 0.8s & 50 & 11 & 12.3s & 6.9s & 83 & 18 & 3m57s & 1m4s & 120 & 26 \\
\hline
P4 & 3.9s & 3.1s & 42 & 9 & 4.5s & 2.1s & 69 & 15 & 1m48s & 1m4s & 100 & 22 \\
\hline
P5 & 5.9s & 2.1s & 63 & 13 & 20.2s & 6.9s & 108 & 21 & 5m12s & 1m4s & 140 & 30 \\
\hline
P6 & 5.3s & 1.4s & 65 & 13 & 17.4s & 6.9s & 108 & 21 & 4m57s & 1m7s & 141 & 30 \\
\hline
P7 & 2m52s & 1m28s & 80 & 15 & 6m13s & 4m6s & 163 & 30 & 9m7s & 2m42s & 231 & 44 \\
\hline
P8 & 3m10s & 1m48s & 77 & 15 & 5m8s & 3m19s & 161 & 30 & 8m47s & 1m44s & 234 & 44 \\
\hline
P9 & 1.2s & 1.1s & 33 & 9 & 2m58s & 2m55s & 58 & 15 & 1m7s & 1m5s & 87 & 22 \\
\hline
P10 & 2m2s & 1m1s & 583 & 146 & 18m37s & 4m46s & 1027 & 249 & 26m51s & 5m2s & 1598 & 372 \\
\hline
P11 & 4m9s & 1m5s & 464 & 112 & 24m41s & 5m25s & 814 & 192 & 46m12s & 18m28s & 1269 & 288 \\
\hline
P12 & 10m33s & 1m5s & 1507 & 370 & 1h4m19s & 5m56s & 2643 & 633 & 2h14m24s & 27m1s & 4116 & 948 \\
\hline
P13 & 17m7s & 3.4s & 14369 & 1088 & 50m5s & 16.1s & 21163 & 1824 & 10h27m40s & 1m58s & 22123 & 2688 \\
\hline
P14 & 17m12s & 3.4s & 14369 & 1088 & 50m0s & 16.0s & 21163 & 1824 & 10h27m47s & 2m21s & 22153 & 2688 \\
\hline
P15 & 17m15s & 3.4s & 14625 & 1088 & 52m32s & 16.4s & 21763 & 1824 & 10h30m26s & 1m35s & 22927 & 2688 \\
\hline
P16 & 16m39s & 3.4s & 14553 & 1088 & 52m10s & 14.9s & 21773 & 1824 & 10h24m23s & 1m27s & 22922 & 2688 \\
\hline
P17 & 5h28m27s & 4m3s & 16066 & 1594 & 19h11m33s & 7m46s & 23568 & 2592 & 17h37m27s & 2m36s & 25236 & 3712 \\
\hline
P18 & 34m42s & 3.7s & 60111 & 3968 & 1h37m9s & 15.5s & 78418 & 6144 & 17h43m2s & 1m54s & 74847 & 8448 \\
\hline
\end{tabular}

\end{center}
\caption{Experiment Results: Compositional.}
\label{tbl:expcomp}
\end{table*}

We have implemented a prototype of the compositional synthesis
algorithm.  The implementation takes as input a finite-data loop-free
C program and converts the program into a Boolean circuit in the
standard way.  We remark that, in principle, a program with
non-input-dependent loops and recursive functions may be converted to
such a form by loop unrolling and function inlining.

The implementation is written in the OCaml programming language.  We
use CIL~\cite{DBLP:conf/cc/NeculaMRW02} for the front-end parsing and
Z3~\cite{DBLP:conf/tacas/MouraB08} for the SMT solver used in the
constraint-based monolithic synthesis.  The experiments are conducted
on a machine with a 2.60GHz Intel Xeon E5-2690v3 CPU with 8GB of RAM
running a 64-bit Linux OS, with the time limit of 20 hours.

We have run the implementation on the 18 benchmark programs taken from
the paper by Eldib and Wang~\cite{DBLP:conf/cav/EldibW14}.  The
benchmarks are (parts of) various cryptographic algorithm
implementations, such as a round of AES, and we refer to their paper
for the details of the respective benchmarks (we use the same program
names).\footnote{Strangely, some of the benchmarks contain inputs
  labeled as ``random'' and have random IO behavior (when those inputs
  are actually treated as randoms), despite the method of
  \cite{DBLP:conf/cav/EldibW14} only supporting programs with
  non-random (i.e., deterministic) IO behavior.  We treat such
  ``random'' inputs as public inputs in our experiments.}  Whereas
their experiments synthesized leakage-resilient versions of the
benchmarks only for the case $n$ is $1$, in our experiments, we do the
synthesis for the cases $n=2$, $n=3$, and $n=4$.

We describe the decomposition strategy that is implemented in the
prototype.  Specifically, we give details of the online decomposition
process mentioned in Section~\ref{sec:decomp}.  The implementation
employs the following {\em adaptive} strategy when decomposing
systematically based on a circuit size bound.  First, we set the bound
to be circuits of some fixed height (the experiments use height $3$),
and decompose based on the bound.  However, in some cases, the bound
can be too large for the monolithic constraint-based synthesis
algorithm to complete in a reasonable amount of time.  Therefore, we
set a limit on the time that the constraint-based synthesis can spend
on constraint solving, and when the time limit is exceeded, we further
decompose that sub-circuit by using a smaller bound.  Further, when
the time limit is exceeded even with the smallest bound, or the number
of secrets in the sub-circuit exceeds a certain bound (this is done to
prevent out of memory exceptions in the SMT solver), we use a pre-made
leakage resilient circuit.  Recall from Remark~\ref{rem:comp} that the
compositionality property ensures the correctness of such a strategy.

Tables~\ref{tbl:expmono} and \ref{tbl:expcomp} summarize the
experiment results.  Table~\ref{tbl:expcomp} shows the results of the
compositional algorithm.  Table~\ref{tbl:expmono} shows the results
obtained by the``monolithic-only'' version of the algorithm.
Specifically, the monolithic-only results are obtained by, first
applying the parallel compositionality property
(cf.~Section~\ref{sec:parcomp}) to divide the given circuit into
separate sub-circuit for each output, and then applying the
constraint-based monolithic synthesis to each sub-circuit and
combining the results.  (The per-output parallel decomposition is
needed because the constraint-based monolithic synthesis only takes
one-output circuits as input -- cf.~Section~\ref{sec:mono}.)  The
monolithic-only algorithm is essentially the monolithic algorithm of
Eldib and Wang~\cite{DBLP:conf/cav/EldibW14} with the improvements
described in Section~\ref{sec:mono}.

We describe the table legends.  The column labeled ``\legname{}''
shows the benchmark program names.  The columns labeled ``\legtime{}''
show the time taken to synthesize the circuit.  Here, ``T/O'' means
that the algorithm was not able to finish within the time limit, and
``M/O'' means that the algorithm aborted due to an out of memory
error.  The columns labeled ``\legsize{}'' show the number of gates in
the synthesized circuit, and the columns labeled ``\legnumrands{}''
show the number of randoms in the synthesized circuit.  

The columns labeled ``\legmaxtpc{}'' in Table~\ref{tbl:expcomp} is the
maximum time spent by the algorithm to synthesizing a sub-circuit in
the compositional algorithm.  Our prototype implementation currently
implements a sequential version of the compositional algorithm where
each sub-circuit is synthesized one at a time in sequence.  However,
in principle, the sub-circuits may be synthesized simultaneously in
parallel, and the columns \legmaxtpc{} give a good estimate of the
efficiency of such a parallel version of the compositional algorithm.
We also remark that the current prototype implementation is
unoptimized and does not ``cache'' the synthesis results, and
therefore, it naively applies the synthesis repeatedly on the same
sub-circuits that have been synthesized previously.

As we can see from the tables, the monolithic-only approach is not
able to finish on many of the benchmarks, even for the case $n = 2$.
In particular, it does not finish on any of the large benchmarks (as
one can see from the sizes of the synthesized circuits, P13 to P18 are
of considerable sizes).  By contrast, the compositional approach was
able to successfully complete the synthesis for all instances.  We
observe that the compositional approach was faster for the larger $n$ in
some cases (e.g., P9 with $n=3$ vs. $n=4$).  While this is partly due
to the unpredictable nature of the back-end SMT solver, it is also an
artifact of the decomposition strategy described above.  More
specifically, in some cases, the algorithm more quickly detects (e.g.,
in earlier iterations of the constraint-based synthesis's CEGAR loop)
that the decomposition bound should be reduced for the current
sub-circuit, which can lead to a faster overall running time.

We also observe that the sizes of the circuits synthesized by the
compositional approach are quite comparable to those of the ones
synthesized by the monolithic-only approach, and the same observation
can be made to the numbers of randoms in the synthesized circuits.  In
fact, in one case (P2 with $n=3$), the compositional approach
synthesized a circuit that is smaller than the one synthesized by the
monolithic-only approach.  While this is due in part to the fact that
the monolithic synthesis algorithm optimizes circuit height rather
than size, in general, it is not inconceivable for the compositional
approach to do better than the monolithic-only approach in terms of
the quality of the synthesized circuit.  This is because the
compositional method could make a better use of the circuit structure
by sharing synthesized sub-circuits, and also because of parsimonious
use of randoms allowed by the compositionality property.  We remark
that the circuits synthesized by our method are orders of magnitude
smaller than those obtained by naively applying the original
construction of Ishai et a.~\cite{DBLP:conf/crypto/IshaiSW03}
(cf.~Theorem~\ref{thm:ishai03}).  For instance, for P18 with $n=4$,
the construction would produce a circuit with more than 3600k gates
and 500k randoms.

\texcomment{P3 has random input-output behavior.  It behaves as not(k1
  AND k2) when r1 and r2 are (1,1), and for other values of r1 and r2,
  it behaves as (k1 AND k2).}

\texcomment{P4 has random input-output behavior.  It behaves as (k1 AND k2) when r1 and r2 are (0,0) or (1,0), and for other values of r1 and r2,
  it behaves as (k1 AND not k2).}

\texcomment{
\begin{table*}[tbp]
\scriptsize
\begin{center}
\begin{tabular}{|l|c|c|c|c|c|c|c|c|c|c|c|c|c|c|c|c|}
\hline
\multicolumn{2}{|c|}{Program} & \multicolumn{6}{|c|}{Monolithic Only} & \multicolumn{9}{|c|}{Compositional}\\
\hline
\multirow{2}{*}{Name} & \multirow{2}{*}{Size} & 
\multicolumn{2}{|c|}{$n=1$} & \multicolumn{2}{|c|}{$n=2$} & \multicolumn{2}{|c|}{$n=3$} & 
\multicolumn{3}{|c|}{$n=1$} & \multicolumn{3}{|c|}{$n=2$} & \multicolumn{3}{|c|}{$n=3$} \\
\cline{3-17}
& & Time & Size & Time & Size & Time & Size 
& Time & Max-Time/Comp & Size & Time & Max-Time/Comp & Size  & Time & Max-Time/Comp & Size \\
\hline
\end{tabular}
\end{center}
\end{table*}
}

\section{Related Work}

\label{sec:related}

\noindent
{\bf Verification and Synthesis for $n$-Threshold-Probing Model of Leakage Resilience.}
The $n$-threshold-probing model of leakage resilience was proposed in
the seminal work by Ishai et al.~\cite{DBLP:conf/crypto/IshaiSW03}.
The subsequent research has proposed methods to build circuits that
are leakage resilient according to the
model~\cite{DBLP:conf/ches/RivainP10,DBLP:conf/fse/CoronPRR13,DBLP:conf/eurocrypt/DucDF14,DBLP:conf/tacas/EldibWS14,DBLP:conf/cav/EldibW14,DBLP:conf/eurocrypt/BartheBDFGS15,DBLP:journals/iacr/BartheBDFG15}.
Along this direction, the two branches of work that are most relevant
to ours are {\em verification} which aims at verifying whether the
given (hand-crated) circuit is leakage
resilient~\cite{DBLP:conf/tacas/EldibWS14,DBLP:conf/eurocrypt/BartheBDFGS15,DBLP:journals/iacr/BartheBDFG15},
and {\em synthesis} which aims at generating leakage resilient
circuits automatically~\cite{DBLP:conf/cav/EldibW14}.  In particular,
our constraint-based monolithic synthesis algorithm is directly
inspired and extends the algorithm given by Eldib and
Wang~\cite{DBLP:conf/cav/EldibW14}.  As remarked before, their method
only handles the case $n = 1$.  By contrast, we propose the first
compositional synthesis approach that also works for arbitrary values
of $n$.

On the verification side, the constraint-based verification method
proposed in \cite{DBLP:conf/tacas/EldibWS14} is a precursor to their
synthesis work discussed above, and it is similar to the candidate
checking phase of the synthesis.  Recent papers by Barthe et
al.~\cite{DBLP:conf/eurocrypt/BartheBDFGS15,DBLP:journals/iacr/BartheBDFG15}
investigate verification methods that aim to also support the case $n
> 1$.  Compositional verification is considered in
\cite{DBLP:journals/iacr/BartheBDFG15}.  As remarked before, in
contrast to the compositionality property described in our paper,
their composition does not require disjointness of the randoms in the
composed components but instead require additional randoms at the site
of the composition.  We believe that the compositionality property
investigated in their work is complementary to ours, and we leave for
future work to combine these facets of compositionality.

We remark that synthesis is substantially harder than verification.
Indeed, in our experience with the prototype implementation, most of
the running time is consumed by the candidate finding part of the
monolithic synthesis process with relatively little time spent by the
candidate checking part.

\vspace{1em}
\noindent
{\bf Quantitative Information Flow.}
{\em Quantitative information flow}
(QIF)~\cite{DBLP:conf/popl/Malacaria07,DBLP:conf/fossacs/Smith09,DBLP:conf/csfw/YasuokaT10,DBLP:conf/csfw/AlvimCPS12}
is a formal measure of information leak, which is based on an
information theoretic notion such as Shannon entropy, R\`{e}nyi
entropy, and channel capacity.  Recently, researchers have proposed to
QIF-based methods for side channel attack
resilience~\cite{DBLP:journals/jcs/KopfB11,DBLP:conf/cav/KopfMO12}
whereby static analysis techniques for checking and inferring QIF are
applied to side channels.  

A direct comparison with the $n$-threshold-probing model of leakage
resilience is difficult because the security ensured by the QIF
approaches are typically not of the form in which the adversary's
capability is restricted in some way, and some (small amount of) leak
is usually permitted.  We remark that, as also observed
by~\cite{DBLP:conf/eurocrypt/BartheBDFGS15}, in the terminology of
information flow theory, the $n$-threshold-probing model of leakage
resilience corresponds to enforcing {\em probabilistic
  non-interference}~\cite{DBLP:conf/sp/Gray91} on every $n$-tuple of
the circuit's internal nodes.

\section{Conclusion}
\label{sec:concl}

We have presented a new approach to synthesizing circuits that are
leakage resilient according to the $n$-threshold-probing model.  We
have shown that the leakage-resilience model admits a certain
compositionality property, which roughly says that composing
$n$-leakage-resilient circuits results in an $n$-leakage-resilient
circuit, assuming the disjointness of the randoms in the composed
circuit components.  Then, by utilizing the property, we have designed
a compositional synthesis algorithm that divides the given circuit
into smaller sub-circuits, synthesizes $n$-leakage-resilient versions
of them containing disjoint randoms, and combines the results to obtain
an $n$-leakage-resilient version of the whole.

\bibliography{lr-lncs}
\bibliographystyle{abbrv}

\addappendix{
\appendix
\section{Proofs}

\subsection{Proof of Theorem~\ref{thm:parcomp}}

\begin{reftheorem}{\ref{thm:parcomp}}
Let $P_1$ and $P_2$ be $n$-leakage-resilient circuits having disjoint
randoms.  Then, $\parcomp{P_1}{P_2}$ is also $n$-leakage-resilient.
\end{reftheorem}
\begin{proof}
  Without loss of generality, we assume that $P_1$ and $P_2$ have the
  same public and secret inputs (otherwise, they can be made so by
  adding dummy connections).  Let $\ve{\rvar_1}$
  (resp.~$\ve{\rvar_2}$) be the randoms in $P_1$ (resp.~$P_2$).  Let
  $\ve{\pvar}$ (resp.~$\ve{\svar}$) be the public (resp.~secret)
  inputs of $\parcomp{P_1}{P_2}$.  Then, the result follows because
  for any $\ve{b_\pvar} \in \aset{0,1}^{|\ve{\pvar}|}$,
  $\ve{b_\svar},\ve{b_\svar}' \in \aset{0,1}^{|\ve{\svar}|}$,
  $\ve{\nd} \in \nodes{\parcomp{P_1}{P_2}}^n$, and $\ve{b_\nd} \in
  \aset{0,1}^n$, we have
  \[
  \begin{array}{l}
    \cnt{\parcomp{P_1}{P_2}}{\ve{b_\pvar}}{\ve{b_\svar}}{\ve{\nd}}{\ve{b_\nd}} = \\
    \hspace{0.3em}
    2^{|\ve{r_2}|}\mult \cnt{P_1}{\ve{b_\pvar}}{\ve{b_\svar}}{\ve{\nd_1}}{\ve{b_{\nd_1}}} + 2^{|\ve{r_1}|}\mult \cnt{P_2}{\ve{b_\pvar}}{\ve{b_\svar}}{\ve{\nd_2}}{\ve{b_{\nd_2}}} =\\
    \hspace{0.3em}
    2^{|\ve{r_2}|}\mult \cnt{P_1}{\ve{b_\pvar}}{\ve{b_\svar'}}{\ve{\nd_1}}{\ve{b_{\nd_1}}} + 2^{|\ve{r_1}|}\mult \cnt{P_2}{\ve{b_\pvar}}{\ve{b_\svar'}}{\ve{\nd_2}}{\ve{b_{\nd_2}}} =\\
    \hspace{0.3em}
    \cnt{\parcomp{P_1}{P_2}}{\ve{b_\pvar}}{\ve{b_\svar}'}{\ve{\nd}}{\ve{b_\nd}}
\end{array}
  \]
  where $\ve{\nd_1}$ and $\ve{\nd_2}$ partition $\ve{\nd}$ to the
  nodes inside $P_1$ and those inside $P_2$, and $\ve{b_{\nd_1}}$
  (resp.~$\ve{b_{\nd_2}}$) is the part of $\ve{b_\nd}$ corresponding
  to the valuations of $\ve{\nd_1}$ (resp.~$\ve{\nd_2}$).
\qed
\end{proof}

\subsection{Proof of Theorem~\ref{thm:seqmultinput}}

For simplicity, we assume that the circuits have no public inputs, and
$P_2$ only has one (non-split) output.  But, the result can be easily
extended to the case with public inputs and the case where $P_2$ has
multiple outputs.  

We first focus on the case where $P_2$ has only one
(non-split) input.  That is, we consider the case where we are given
$P_1$ that takes $x$ many inputs and $P_2$ that takes one input.
Therefore, the goal is to show the following.
\begin{theorem}
\label{thm:seqoneinput}
Let $P_1$ be an $x$-input $n$-leakage-resilient circuit and $P_2$ be
an $1$-input $n$-leakage-resilient circuit, having disjoint randoms.
Then, $\compcirc{P_1}{P_2}$ is $n$-leakage-resilient.
\end{theorem}
The general case where $P_2$ takes multiple inputs follows directly
from the proof of this case and is deferred to
Section~\ref{sec:multinputp2}.

We begin by introducing preliminary definitions and notations that are
used in the proof.  We represent circuits by their {\em truth tables}
(or, simply {\em tables}) whose columns are the (observable) nodes of
the circuit and rows are the possible valuation of the nodes.
Formally, for an $x$-input circuit $P$ that is \plusone{n}-share
split, the table of $P$, $\tbl(P)$, has its first $(n+1) \mult x$
columns correspond to the input nodes,\footnote{For simplicity, we assume that each secret (non-split) input in $P_1$ and $P_2$ is split into shares by a single input encoder.  The proof can be adopted to the case where a secret is encoded by multiple input encoders.} followed by some number of
intermediate nodes (which may include additional randoms), and
followed by $n+1$ columns corresponding to the output nodes.  We write
$\rows{t}$ (resp.~$\cols{t}$) for the number of rows (resp.~columns)
of the table $t$.  We write $t[i]$ for the $i$th row of the table, and
$t(i,j)$ for the value of the $(i,j)$th entry of $t$.

\begin{example}
\label{ex:tbl}
\em Below is a table for the \plusone{1}-share split $1$-input
circuit.  The circuit has 6 nodes $\nd_1,\dots,\nd_6$ where the
(split) input nodes are $\nd_1,\nd_2$, the intermediate nodes are
$\nd_3,\nd_4$ such that $\nd_3$ is an additional random node, and the
(split) output nodes are $\nd_5,\nd_6$.  The circuit is wired so that
$\nd_4 = \nd_1 \oplus 1$, $\nd_5 = \nd_4 \oplus \nd_3$, and $\nd_6 =
\nd_2 \oplus \nd_3$.  That is, the circuit outputs the XOR of the
input with 1 (i.e., it is a negation circuit).
\[
\begin{array}{|c:c:c:c:c:c|}
\hline
\nd_1 & \nd_2 & \nd_3 & \nd_4 & \nd_5 & \nd_6 \\
\hline
0 & 0 & 0 & 1 & 1 & 0 \\\hdashline
0 & 0 & 1 & 1 & 0 & 1 \\\hdashline
1 & 1 & 0 & 0 & 0 & 1 \\\hdashline
1 & 1 & 1 & 0 & 1 & 0 \\\hdashline
1 & 0 & 0 & 0 & 0 & 0 \\\hdashline
1 & 0 & 1 & 0 & 1 & 1 \\\hdashline
0 & 1 & 0 & 1 & 1 & 1 \\\hdashline
0 & 1 & 1 & 1 & 0 & 0 \\
\hline
\end{array}
\]
\end{example}

\subsubsection{Leakage Resilience as Table Safety}

The notion of $n$-leakage-resilience can be formalized as a certain
property of tables, which we call {\em safety}.  Informally, a table
being $n$ safe means that, the number of rows of the table
corresponding to each (non-split) input remains the same even after
reducing the table $n$ times by choosing a column and a bit and
removing the rows that do not have that bit at that column.  We
formalize the notion below.

Let $\ve{b} \in \aset{0,1}^x$ and $m > 0$.  For $\ve{a} \in
\aset{0,1}^*$ of length at least $x\mult m$, we write
$\ve{a}\match{x}{m}\ve{b}$ if for each $i \in \aset{1, \dots, x}$,
$b_i = \bigoplus_{(i-1)\mult m < j \leq i\mult m} a_j$.  Let $t$ be a
table such that $\cols{t} \geq x\mult m$.  The {\em $\ve{b}$
  restriction of $t$ wrt $m$}, written $\brest{\ve{b}}{m}{t}$, is
defined to be a table comprising the rows $i$ of $t$ such that
$t[i]\match{x}{m}\ve{b}$.  Intuitively, $\brest{\ve{b}}{n+1}{t}$ is
the sub-table of $t$ comprising the rows with the non-split input
$\ve{b}$ (for the \plusone{n} split case).  For example, let $t$ be
the table from Example~\ref{ex:tbl}, then $\brest{0}{2}{t}$ is the
table comprising the first four rows of $t$, and $\brest{1}{2}{t}$ is
the table comprising the remaining (i.e., the last four) rows of $t$.

For $(j,b) \in \aset{1,\dots,\cols{t}}\times\aset{0,1}$, we write $t'
\reduce{j}{b} t$ if $t'$ is $t$ but with the rows having the value
$\neg b$ at the $j$th column removed.  Intuitively, $t \reduce{j}{b}
\tbl(P)$ means that $t$ is the state of the circuit $P$ after the
adversary observed the value $b$ at the $j$th node.  For $m \geq 0$, we
write $t' \reductions{m} t$ if $t'$ is obtained by $m$ adversary
observations from $t$, that is, if there exist
$(j_1,b_1),\dots,(j_m,b_m)$ such that $t = t_1$, $t' = t_m$, and
$t_{i+1} \reduce{j_i}{b_i} t_i$ for each $i \in \aset{1,\dots,m-1}$.

Let $m \geq 0$ and $\ell > 0$.  We define the notion {\em
  $\safe{x}{m}{\ell}$ safe} inductively as follows.  We say that a
table $t$ is $\safe{x}{0}{\ell}$ safe if for any $\ve{b},\ve{b}' \in
\aset{0,1}^x$, $\rows{\brest{\ve{b}}{\ell}{t}} =
\rows{\brest{\ve{b}'}{\ell}{t}}$.  For $m > 0$, $t$ is
$\safe{x}{m}{\ell}$ safe if for any $t'$ such that $t' \reductions{1}
t$, $t'$ is $\safe{x}{m-1}{\ell}$ safe.  Note that an empty table is
$\safe{x}{m}{\ell}$ safe for any $x$, $m$ and $\ell$.  Also, it is
immediate from the definition that safety is {\em monotonic}, that is,
if $t$ is $\safe{x}{m_1}{\ell}$ safe and $m_1 \geq m_2$ then $t$ is
also $\safe{x}{m_2}{\ell}$ safe.\texcomment{Proof: 1->0: Suppose not 0
  safe show not 1 safe: if any column has the same number of 0's and
  1's, then just pick it table don't change and it's not 1 safe.
  Otherwise, pick the column having diff number of 0's and 1's and not
  1 safe.  n->n-1 for n>0 is trivial.}  Also, the following is
immediate from the definition.
\begin{lemma}
\label{lem:safety}
Suppose $m_1 \leq m_2$.  Then, $t$ is $\safe{x}{m_2}{\ell}$ safe iff
for any $t'$ such that $t' \reductions{m_1} t$, $t'$ is
$\safe{x}{m_2-m_1}{\ell}$ safe.
\end{lemma}

It is easy to see that the notion of safety captures leakage
resilience.  That is, we have the following.
\begin{lemma}
\label{lem:tbllrgen}
The following are equivalent.
\begin{itemize}
\item[(1)] $x$-input \plusone{n}-split $P$ is $n$-leakage-resilient.
\item[(2)] $\tbl(P)$ is $\safe{x}{n}{n+1}$ safe.
\item[(3)] For any $t$ such that $t \reductions{n} \tbl(P)$, $t$ is
  $\safe{x}{0}{n+1}$ safe.
\end{itemize}
\end{lemma}
\begin{proof}
  It is not hard to see that (2) is equivalent to
  $n$-leakage-resilience.  (2) $\Leftrightarrow$ (3) follows from
  Lemma~\ref{lem:safety}.
\end{proof}

\subsubsection{Circuit Composition as Table Composition}

We capture circuit composition by composition of tables.  More
precisely, we shall define the table composition operation
$\comptbl{}{}_m$ so that the composition of the tables
$\comptblm{\tbl(P_1)}{\tbl(P_2)}{n+1}$ is the table of the composed
circuit $\tbl(\compcirc{P_1}{P_2})$.

For $m > 0$, let us write $\hd{\ve{b}}{m}$ (resp.~$\tl{\ve{b}}{m}$)
for the first (resp.~last) $m$ elements of $\ve{b}$.  For $t_1$ and
$t_2$, and $m > 0$ such that $m \leq \cols{t_1}$ and $m \leq
\cols{t_2}$, we define $\comptblm{t_1}{t_2}{m}$ to be the
$\cols{t_1}+\cols{t_2} - m$ column table such that each $\ve{b} \in
\aset{0,1}^{\cols{t_1}+\cols{t_2} - m}$ appears $i_1 \mult i_2$ times
in the table where $i_1$ (resp.~$i_2$) is the number of times
$\hd{\ve{b}}{\cols{t_1}}$ appears in $t_1$
(resp.~$\tl{\ve{b}}{\cols{t_2}}$ appears in $t_2$).  (Therefore,
$\ve{b}$ occurs as a row of $\comptblm{t_1}{t_2}{m}$ iff
$\hd{\ve{b}}{\cols{t_1}}$ occurs as a row of $t_1$ and
$\tl{\ve{b}}{\cols{t_2}}$ occurs as a row of $t_2$.)  For $t_1$ and
$t_2$ such that $\cols{t_1} = \cols{t_2} = m$, we simply write
$\comptbl{t_1}{t_2}$ for $\comptblm{t_1}{t_2}{m}$.  Note that
$\comptbl{t_1}{t_2} = \comptbl{t_2}{t_1}$.

\texcomment{
\begin{example}
\em
Example of table composition \todo

\end{example}
}

It can be seen that the table composition correctly captures circuit
composition.  That is, for \plusone{n}-split circuits, we have
$\tbl(\compcirc{P_1}{P_2}) = \comptblm{\tbl(P_1)}{\tbl(P_2)}{n+1}$.
Also, it is immediate from the definitions that any attack on the
composed table can be simulated by an attack on the components.
More formally, we have the following.
\begin{lemma}
\label{lem:decomp}
Let $m,\ell > 0$.  Then, $t \reductions{m} \comptblm{t_1}{t_2}{\ell}$
iff there exist $m_1$, $m_2$, $t_1'$, and $t_2'$ such that $m_1+m_2 =
m$, $t_1' \reductions{m_1} t_1$, $t_2' \reductions{m_2} t_2$, and $t =
\comptblm{t_1'}{t_2'}{\ell}$.
\end{lemma}

We define the notion of {\em deterministic} tables.  Roughly a table
being deterministic means that split input valuations for the same
non-split input leads to the same non-split output.  More formally,
for $m > 0$ and $\ve{b} \in \aset{0,1}^x$ let
$\splitouts{m}{\ve{b}}{t}$ be the $m$ column subtable of $t$
comprising the last $m$ columns of rows $i$ such that
$t[i]\match{x}{m}\ve{b}$.  We call $\splitouts{m}{\ve{b}}{t}$ the {\em
  output table} for the non-split input $\ve{b}$.  We write
$\nonsplitvals{t}$ for the $1$-column $\rows{t}$-row table $t'$ such
that for each row $i$, $t[i]\match{1}{\cols{t}}t'[i]$.  We say that
$t$ is {\em $\dettbl{x}{m}$ deterministic} if for each $\ve{b} \in
\aset{0,1}^x$, the elements of
$\nonsplitvals{\splitouts{m}{\ve{b}}{t}}$ are either all $0$ or all
$1$.  Because any $n$-leakage-resilient circuit with which we are
concerned is IO-equivalent to a random-free circuit, without loss of
generality, we may assume that the table $\tbl(P_1)$ (resp.~$\tbl(P_2)$)
is $\dettbl{x}{n+1}$ (resp.~$\dettbl{1}{n+1}$) deterministic.

Our goal is to show Theorem~\ref{thm:seqoneinput}.  That is, given
$n$-leakage-resilient circuits $P_1$ and $P_2$, the composed circuit
$\compcirc{P_1}{P_2}$ is also $n$-leakage-resilient.  By
Lemmas~\ref{lem:tbllrgen}, \ref{lem:safety}, and \ref{lem:decomp}, it
follows that it suffices to show that following.
\begin{theorem}
\label{thm:main}
Let $m_1 > 0$ and $m_2 > 0$ such that $m_1 + m_2 = n$.  Let $t_1$ and
$t_2$ be tables such that $t_1$ is $\dettbl{x}{n+1}$ deterministic and
$\safe{x}{m_1}{n+1}$ safe, and $t_2$ is $\dettbl{1}{n+1}$
deterministic and $\safe{1}{m_2}{n+1}$ safe.  Then,
$\comptblm{t_1}{t_2}{n+1}$ is $\safe{x}{0}{n+1}$ safe.
\end{theorem}

We prove the theorem in the following subsection.

\subsubsection{Proof of Theorem~\ref{thm:main}}

First, we prove a few technical lemmas.  We write $t_1 \eqtbl t_2$ if
$t_1$ and $t_2$ are equivalent up to permutation of the rows.  The
next lemma shows that, for a table whose rows all sum up to a given
value, the numbers of $1$'s and $0$'s in each column uniquely
determine the table up to row permutation.
\begin{lemma}
\label{lem:01batch}
Let $\ve{v}$ be a vector of non-negative integers of length $m$.  Let
$b \in \aset{0,1}$. Let $t_1$ and $t_2$ be $m$ column tables with
$\rows{t_1} = \rows{t_2}$ such that for each $t_\ell$ ($\ell \in
\aset{1,2}$), 1.) for every row $\ve{b} \in t_\ell$, $b =
\bigoplus_{1\leq i \leq \cols{t}} b_i$ (i.e., every row sums up to the
same value), and 2.) for each column $j \in \aset{1,\dots,m}$, there
are exactly $v_j$ number of $1$'s in the $j$th column of $t_\ell$ (and
hence $\rows{t_\ell} - v_j$ number of $0$'s).  Then, $t_1 \eqtbl t_2$.
\end{lemma}
\begin{proof}
  We show that $b$ and $\ve{v}$ uniquely determines the table content
  (up to row permutation) for tables whose row size and column size
  are fixed.  We prove by induction on the column size, $m$.  The case
  $m = 1$ is immediate, since in this case either the table is all
  $0$'s or all $1$'s.

  We show the inductive case.  Let $m > 1$ and $t$ be a table
  satisfying the requirements 1.) and 2.) for $b$ and $\ve{v}$.  Let
  $t_0$ (resp.  $t_1$) be $m-1$ column subtable of $t$ comprising the
  columns $2$, \dots, $m$ of $t$, but only taking the rows $i$ such
  that $t(i,1) = 0$ (resp.~$t(i,1) = 1$).  For each $i \in \aset{2,\dots,m}$,
  let $v_{0i}$ be the number of times $1$ occurs in the $i-1$st column
  of $t_0$ and let $v_{1i} = v_i - v_{0i}$.  Let $\ve{v}_0 =
  v_{02},\dots,v_{0m}$ and $\ve{v}_1 = v_{12},\dots,v_{1m}$. Then,
  $t_0$ (resp.~$t_1$) satisfies the requirement with $b$ (resp.~$\neg
  b$) and $\ve{v}_0$ (resp.~$\ve{v}_1$).  By induction hypothesis,
  $t_0$ and $t_1$ are uniquely determined, and therefore, so is $t$.
\end{proof}
We use the above lemma to show the following, which says that, if a
table is at least $1$ safe, then non-split inputs that map to the
same non-split output value have the same split output valuations.
\begin{lemma}
\label{lem:outsame}
Let $t$ be $\dettbl{x}{m}$ deterministic and $\safe{x}{1}{m}$ safe.
Let $\ve{b},\ve{b}' \in \aset{0,1}^x$.  Let $t_1 =
\splitouts{m}{\ve{b}}{t}$ and $t_2 = \splitouts{m}{\ve{b}'}{t}$.
Then, either $t_1 \eqtbl t_2$, or $\nonsplitvals{t_1}$ and
$\nonsplitvals{t_2}$ have no common elements.
\end{lemma}
\begin{proof}  
  Note that, because of $1$ safety, for each $j \in \aset{1,\dots,m}$,
  the number of $1$'s and $0$'s in the $j$th column of $t_1$ must be
  the same as those of $t_2$.  Suppose $\nonsplitvals{t_1}$ and
  $\nonsplitvals{t_2}$ have a common element.  Then, by determinism,
  every row of $t_1$ and $t_2$ sums up to a same value.  Therefore, by
  Lemma~\ref{lem:01batch}, we have $t_1 \eqtbl t_2$.
\end{proof}

Next, we prove a couple of lemmas that focus on the case where the
tables considered are $n+1$ column tables that are $\safe{1}{\_}{n+1}$
safe.  The first lemma, Lemma~\ref{lem:nsafetbl}, is used to prove
Lemma~\ref{lem:main}.  Lemma~\ref{lem:main} is used in a part of the
proof of Theorem~\ref{thm:main} that focuses on the output columns of tables
obtained from $P_1$ and the input columns of the tables obtained from
$P_2$, which will be $n+1$ column tables that are $\safe{1}{\_}{n+1}$
safe.
\begin{lemma}
\label{lem:nsafetbl}
Let $t$ be a table with $\cols{t} = n+1$.  Then, $t$ is $\safe{1}{n}{n+1}$
safe iff 1.) $t$ contains every element of $\aset{0,1}^{n+1}$ as a
row, and 2.) each element of $\aset{0,1}^{n+1}$ appears the same
number of times in $t$.
\end{lemma}
\begin{proof}
  The if direction is trivial.  We show the only if direction.
  Suppose that $t$ is $\safe{1}{n}{n+1}$ safe.  We show that 1.)  and
  2.) must be satisfied.

  First, we show that 1.) is satisfied.  Let $\ve{b}_0$ be a row of
  $t$.  Then, any $\ve{b}'$ that differs from $\ve{b}_0$ at one entry
  must also appear as a row of $t$ since otherwise the adversary can
  distinguish the 0-restriction and 1-restriction after observing all
  $n$ columns except for that entry.  Because any $\ve{b} \in
  \aset{0,1}^{n+1}$ can be reached from $\ve{b}_0$ via a series of
  one-entry changes, any such $\ve{b}$ must appear as a row of $t$.

  Finally, note that by the same reasoning as above, for any row
  $\ve{b}$ of $t$, any $\ve{b}'$ different from $\ve{b}$ in one entry
  must appear the same number of times as $\ve{b}$.  Hence, 2.)  is
  also satisfied.
\end{proof}
For a table $t$ and $j \in \aset{1,\dots,\cols{t}}$, we say that the
column $j$ is {\em fixed} in $t$ if there exists $b \in \aset{0,1}$
such that for any $i \in \aset{1,\dots,\rows{t}}$, $t(i,j) = b$ (i.e.,
the column is all $b$).  Note that a table is $\safe{1}{m}{\_}$ safe only
if there are at least $m$ unfixed columns.  For a table $t$ and $S
\subseteq \aset{1,\dots,\cols{t}}$, we write $t\rest_S$ for the table
comprising the columns $S$ of $t$.
\begin{lemma}
\label{lem:main}
Let $m_1 \geq 0$, $m_2 \geq 0$ and $n = m_1 + m_2$.  Let $t_1$ and
$t_2$ be $n+1$ column tables such that $t_1$ (resp.~$t_2$) is
$\safe{1}{m_1}{n+1}$ safe (resp.~$\safe{1}{m_2}{n+1}$ safe).  Then,
$\comptbl{t_1}{t_2}$ is $\safe{1}{0}{n+1}$ safe.
\end{lemma}
\begin{proof}
  We prove by simultaneous induction on $n$, $m_1$ and $m_2$.  The
  base case where $n=0$, $m_1 = 0$, or $m_2 = 0$ follows from
  Lemma~\ref{lem:nsafetbl} because $m_2 = n$ or $m_1 = n$ in this case.
  \texcomment{For the induction on $n$, the base case $n=1$ follows
    since either $m_1 = 0$ or $m_2 = 0$ in this case.}

  We show the inductive case.  Firstly, consider the case the fixed
  columns of $t_1$ and $t_2$ are the same.  Without loss of
  generality, we assume that the fixed columns have the same values in
  $t_1$ and in $t_2$ since otherwise $\comptbl{t_1}{t_2}$ is empty and
  the result is immediate.  Let $t_1'$ (resp.~$t_2'$) be the sub-table
  of $t_1$ (resp.~$t_2$) comprising the unfixed columns.  Let $k$ be
  the number of fixed columns.  Then, $t_1'$ (resp.~$t_2'$) is
  $\safe{1}{m_1-k}{n-k+1}$ (resp.~$\safe{1}{m_2-k}{n-k+1}$) safe (in
  fact, they are $m_1$ and $m_2$ safe).  Therefore, by induction
  hypothesis, $\comptbl{t_1'}{t_2'}$ is $\safe{1}{0}{n-k+1}$ safe.
  Because $\comptbl{t_1}{t_2}$ is $\comptbl{t_1'}{t_2'}$ but with the
  $k$ fixed columns of $t_1$ (or $t_2$) added, $\comptbl{t_1}{t_2}$ is
  $\safe{1}{0}{n+1}$ safe.
  
  Continuing with the inductive case, we now consider the case the
  fixed columns of $t_1$ and $t_2$ are not the same.  Suppose that
  there exists a column that is fixed in $t_2$ but not in $t_1$ (the
  argument is similar for the case when there is a column that is
  fixed in $t_1$ but not in $t_2$).  Let $j \in \aset{1,\dots,n+1}$ be
  fixed in $t_2$ but not in $t_1$.  Let $b$ be the entry of the $j$th
  column of $t_2$, and let $t_1' \reduce{j}{b} t_1$.  Let $S =
  \aset{1,\dots,n+1}\setminus \aset{j}$.  Then, $t_1'\rest_{S}$ is
  $\safe{1}{m_1-1}{n}$ safe, and $t_2\rest_{S}$ is $\safe{1}{m_2}{n}$
  safe.  Therefore, by induction hypothesis,
  $\comptbl{t_1'\rest_{S}}{t_2\rest_{S}}$ is $\safe{1}{0}{n}$ safe.
  Because $\comptbl{t_1'}{t_2}$ is
  $\comptbl{t_1'\rest_{S}}{t_2\rest_{S}}$ but with an additional all
  $b$ column, $\comptbl{t_1'}{t_2}$ is $\safe{1}{0}{n+1}$ safe.
  Finally, because $\comptbl{t_1}{t_2} = \comptbl{t_1'}{t_2}$, the
  result follows.
\end{proof}

We are now ready to prove Theorem~\ref{thm:main}, restated here.

\begin{reftheorem}{\ref{thm:main}}
  Let $m_1 > 0$ and $m_2 > 0$ such that $m_1 + m_2 = n$.  Let $t_1$
  and $t_2$ be tables such that $t_1$ is $\dettbl{x}{n+1}$
  deterministic and $\safe{x}{m_1}{n+1}$ safe, and $t_2$ is
  $\dettbl{1}{n+1}$ deterministic and $\safe{1}{m_2}{n+1}$ safe.
  Then, $\comptblm{t_1}{t_2}{n+1}$ is $\safe{x}{0}{n+1}$ safe.
\end{reftheorem}
\begin{proof}
  Let $t_1^{\it out}$ be the last $n+1$ columns of $t_1$.  If all rows
  of $t_1^{\it out}$ sum up to the same value, then by
  Lemma~\ref{lem:outsame}, the result follows.

Otherwise, it means that not every row of $t_1$ maps to the same
non-split output.  Then, by Lemma~\ref{lem:outsame} again, for each
non-split output valuation (i.e., $0$ or $1$), the output tables are
equivalent for all inputs that maps to the same output.

Let $\ve{b}, \ve{b}' \in \aset{0,1}^x$ be non-split inputs of $t_1$
that map to different outputs.  That is, there exist rows $i$ and $j$
of $t_1$ with $t_1[i]\match{x}{n+1}\ve{b}$,
$t_1[j]\match{x}{n+1}\ve{b}'$, and $(\bigoplus_{1 \leq \ell \leq n+1}
a_\ell) \neq (\bigoplus_{1 \leq \ell \leq n+1} a'_\ell)$ where $\ve{a}
= \tl{t_1[i]}{n+1}$ and $\ve{a}' = \tl{t_1[j]}{n+1}$.  Let $t_1'$ be
the table comprising the rows of $\splitouts{n+1}{\ve{b}}{t_1}$ and
$\splitouts{n+1}{\ve{b}'}{t_1}$.  Let $t_2^{\it in}$ be the table
comprising the first $n+1$ columns of $t_2$.  Then, for any $m \geq
0$, to show $\comptblm{t_1}{t_2}{n+1}$ is $\safe{x}{m}{n+1}$, it
suffices to show that $\comptbl{t_1'}{t_2^{\it in}}$ is
$\safe{1}{m}{n+1}$ safe (for arbitrary such $t_1'$ constructed from
some $\ve{b}, \ve{b}' \in \aset{0,1}^x$ mapping to different outputs).
Because $t_1$ is $\safe{x}{m_1}{n+1}$ safe, $t_1'$ is
$\safe{1}{m_1}{n+1}$ safe.  Also, because $t_2$ is
$\safe{1}{m_2}{n+1}$ safe, $t_2^{\it in}$ is $\safe{1}{m_2}{n+1}$
safe.  Therefore, by Lemma~\ref{lem:main}, $\comptbl{t_1'}{t_2^{\it
    in}}$ is $\safe{1}{0}{n+1}$ safe, and the result follows.
\end{proof}

\begin{remark}
  \em In the definition of sequential composition, we have left
  unspecified which wire of the $P_1$'s split output is connected to
  which wire of the $P_2$'s split input.  Because the table joining
  operation $\comptblm{\tbl(P_1)}{\tbl(P_2)}{n+1}$ is insensitive to
  how the rows of $\tbl(P_1)$ and $\tbl(P_2)$ are permuted, the proof
  shows that the compositionality result holds regardless of the
  connection choice.
\end{remark}

\subsubsection{Case $P_2$ Has Multiple Inputs}

\label{sec:multinputp2}

We now extend the sequential compositionality result to the case $P_2$ has
multiple inputs.  In this case, we have $y$ many circuits
$P_{11},P_{12},\dots,P_{1y}$ whose outputs will be connected to $P_2$,
and we denote the composed circuit as
$\compcirc{P_{11},P_{12},\dots,P_{1y}}{P_2}$.  The goal is to show the
following.
\begin{reftheorem}{\ref{thm:seqmultinput}}
Let $P_{11},\dots,P_{1y}$ be $n$-leakage-resilient circuits, and $P_2$
be an $y$-input $n$-leakage-resilient circuit, having disjoint
randoms.  Then, $\compcirc{P_{11},\dots,P_{1y}}{P_2}$ is
$n$-leakage-resilient.
\end{reftheorem}
The case follows immediately from the proof shown above for the case
where $P_2$ only has one input.  This is because, for $n=m_1+m_2$ and
each $P_{1i}$, Theorem~\ref{thm:main} implies that attacking $P_{1i}$
$m_2$ times and $P_2$ $m_1$ times and composing them cannot
distinguish the secret inputs given to $P_{1i}$.  Also, because of the
disjointness of the randoms in the components, attacks on $P_{1j}$ for
$j\neq i$ has no effect on the table of $P_{1i}$.  Then, because $m_2$
safety of $P_2$ implies $m_2$ safety of the portion of $\tbl(P_2)$
corresponding to the $i$th input, the result follows.  We remark
that the proof actually shows that the composition $\compcirc{P_{11},\dots,P_{1y}}{P_2}$ can withstand an attack where the adversary observes
$m_2$ nodes of $P_2$ and $m_{1i}$ nodes of each $P_{1i}$ such that
$n \geq m_2 + \max_i (m_{1i})$.

}{}

\end{document}